\documentclass[runningheads]{llncs}
\usepackage{changes}
\usepackage{graphicx}
\usepackage{multirow}
\usepackage{amsmath,amsfonts,amssymb}
\usepackage{float}
\usepackage{array}
\usepackage{booktabs}
\usepackage{wrapfig}
\usepackage[misc]{ifsym}
\usepackage[ruled,vlined,linesnumbered]{algorithm2e}
\usepackage[colorlinks,
            linkcolor=blue,
            anchorcolor=blue, 
            citecolor=blue, 
            ]{hyperref}
            
\newcolumntype{P}[1]{>{\raggedleft\arraybackslash}p{#1}}

\usepackage{stmaryrd}
\usepackage{enumerate}
\usepackage{mathtools}
\usepackage{ulem}
\usepackage{orcidlink}
\usepackage{subcaption}
\usepackage{bm}
\usepackage{enumitem} 

\allowdisplaybreaks

\newcommand{\inferenceOne}[3]{
    \dfrac{#3}{#2}
    \hspace{0.8em} \textbf{#1}
}

\newcommand{\myHspace}{
    \hspace{1.2em}
}

\begin{document}

\title{Compositional Abstraction for Timed Systems with Broadcast Synchronization\thanks{This is the extended version of our paper accepted at CAV 2025}}

\titlerunning{Compositional Abstraction for Timed Systems with Broadcast Sync.}

\author{
    Hanyue Chen\inst{1}\,\orcidlink{0000-0002-9434-1695} \and
    Miaomiao Zhang\inst{1}\textsuperscript{(\Letter)} \and
    Frits Vaandrager\inst{2}
} 
\authorrunning{H. Chen et al.}
%
\institute{Tongji University, Shanghai, China\\
\email{\{2111285,miaomiao\}@tongji.edu.com}\\ \and
Radboud University, Nijmegen, The Netherlands  \\
\email {F.Vaandrager@cs.ru.nl}
}
\maketitle              
\begin{abstract}
    Simulation-based compositional abstraction effectively mitigates state space explosion in model checking, particularly for timed systems. However, existing approaches do not support broadcast synchronization, an important mechanism for modeling non-blocking one-to-many communication in multi-component systems. Consequently, they also lack a parallel composition operator that simultaneously supports broadcast synchronization, binary synchronization, shared variables, and committed locations. To address this, we propose a simulation-based compositional abstraction framework for timed systems, which supports these modeling concepts and is compatible with the popular U{\scriptsize PPAAL} model checker. Our framework is general, with the only additional restriction being that the timed automata are prohibited from updating shared variables when receiving broadcast signals. Through two case studies, our framework demonstrates superior verification efficiency compared to traditional monolithic methods.
\end{abstract}

\section{Introduction}\label{sec: introduction}
\emph{Model checking}~\cite{baier2008principles,1997Model,merz2000model,queille1982specification} is a widely used technique for automatically verifying whether a system meets specified properties by exploring its state space. However, in real-world systems, especially timed systems, the state space size grows exponentially with the number of components~\cite{clarke1981design}. This leads to the state space explosion problem, making exploring and storing the states harder for verification. \emph{Simulation-based compositional abstraction}~\cite{MilnerSimulation} is a recognized method to address this issue~\cite{Compositional_Abstraction}, which simplifies systems by replacing complex components with abstractions that preserve essential behaviors, reducing the state space and improving verification efficiency. 

Models of the systems have different communication mechanisms, such as reading and writing \emph{shared variables} in TLA and TLA+~\cite{TLA_Lamport1994}, \emph{binary synchronization} through paired input/output actions in CCS~\cite{CCS} and Pi-Calculus~\cite{Pi_calculate}, etc. Among them, \emph{broadcast synchronization} is also important, where a sender emits a synchronization signal in a non-blocking manner to multiple receivers, each deciding independently whether to accept the signal. The non-blocking nature allows for unified and concise modeling of synchronization among multiple components, as it imposes no limit on the number of actual receivers.
For the widely used model-checking tool U{\scriptsize PPAAL}~\cite{uppaaltutorial,Larsen2022}, it supports broadcast synchronization and has been applied successfully in various industrial cases~\cite{zeroconf,Upp_broadcast2,Upp_broadcast3,Upp_broadcast1,Using_Uppaal_1,new_case_try,Upp_broadcast5}. 

Several simulation-based compositional abstraction frameworks address synchronization involving multiple participants. For instance, the framework in~\cite{Sociable_Interfaces} prohibits receivers from updating shared variables to support one-to-one, one-to-many, and many-to-many synchronization. However, this framework is designed for untimed systems, and its compositional rules are not associative~\cite{Berendsen2007_sep}. The frameworks in~\cite{Compositional_Interchange,SZPAK2020229} integrate shared variables with \emph{multi-cast synchronization} in timed systems. They achieve associative compositional rules by requiring synchronized transitions to update shared variables simultaneously, ensuring consistent valuations before and after synchronization. Nevertheless, this kind of synchronization needs all components with synchronized actions to participate, which conflicts with the non-blocking nature of broadcast synchronization. 
A specification framework for real-time systems based on timed input/output automata (TIOAs) is developed in~\cite{TIOA_Kim}, but the synchronization between TIOAs is also a kind of multicast synchronization.
To our knowledge, no simulation-based compositional abstraction framework currently exists to handle timed models with broadcast synchronization.
Although it is possible to emulate broadcast synchronization in terms of other communication mechanisms, e.g., binary synchronization, this often introduces additional states, reducing the naturalness and readability of the models~\cite{Broadcasting_Embedded_Systems}, and bringing difficulty to verification.

Furthermore, designing a simulation-based compositional abstraction framework that supports multiple communication mechanisms for timed systems is necessary. Several efforts have been made on compositional verification for the composed models with shared variables and binary synchronization. For example, the framework in~\cite{Hybrid_I/O_automata} restricts that a shared variable can only be updated in the same automaton. The work in~\cite{Fault-Tolerant_Systems} relaxes the restriction by allowing the update of shared variables in multiple automata through internal transitions, whereas it still does not support the update during synchronization. 
The framework proposed by Berendsen and Vaandrager removes this restriction, supporting binary synchronization, shared variables, and \emph{committed locations}~\cite{Compositional_Abstraction}, which is one of the key features of U{\scriptsize PPAAL}, ensuring atomic transitions and significantly reducing the state space by excluding irrelevant behaviors~\cite{committedness}. As known, if a system component is in a committed location, time cannot progress, and the next system transition must start from that location. This framework has been successfully used to verify the Zeroconf protocol for any number of hosts~\cite{zeroconf} but lacks support for compositional abstraction with broadcast synchronization.


Hence, we aim to develop a simulation-based compositional abstraction framework for timed systems with broadcast synchronization. Given that U{\scriptsize PPAAL} offers a rich syntax for modeling complex systems as \emph{networks of timed automata} (NTAs)~\cite{alur1999timed,alur1994theory}, this framework is also designed to support binary synchronization, shared variables, and committed locations.



To achieve this, first inspired by the definition of \emph{timed transition systems} (TTSs) in~\cite{Compositional_Abstraction}, we introduce the \emph{timed transition systems with broadcast actions} (TTSBs), which extend \emph{labeled transition systems} (LTSs) with state variables, transition commitments, and time-related behaviors. When combining broadcast synchronization with shared variables, the order in which these variables are updated by the transitions involved in the synchronization can lead to different system states. Therefore, we prohibit TTSBs from updating shared variables when receiving broadcast signals, which is crucial for proving the theorem that the parallel composition we designed for TTSBs is both commutative and associative.
Next, we introduce a CCS-style restriction operator to internalize a set of synchronized actions and shared variables so that no further TTSBs may communicate via them. This restriction is useful, as multiple component models might be abstracted into a single model, and these actions and variables should be considered as internal in the abstraction.


Secondly, considering that timed systems are often modeled as NTAs, we give two kinds of semantics of NTAs for subsequent compositional abstraction. The first one strictly follows the U{\scriptsize PPAAL} semantics, which directly transforms an entire NTA of a timed system into an LTS. 
However, this semantics lacks compositionality, which makes it impossible to abstract parts of the system model, that is, to replace one or more components with simpler ones. We refer to it as non-compositional semantics.
So we define the second one, compositional semantics for NTAs, achieved by converting each \emph{timed automaton} (TA) into its corresponding TTSB, composing them in parallel, applying restriction operations, and extracting the underlying LTS. 
We further prove the theorem that these two semantics are equivalent, laying the foundation for subsequent compositional abstraction.


Thirdly, since the compositional semantics of an NTA are derived based on a TTSB, the abstraction relations between NTAs can be defined in terms of the relation between TTSBs. We describe a timed step simulation relation of TTSBs and prove the theorem that this relation is a precongruence for parallel composition. 
This allows the system to be abstracted by replacing one or more components with simpler models that preserve essential behaviors. For example, abstracting multiple consumers in a producer-consumer system into a single simplified model.


Finally, based on the previous theorems 
we prove that if the abstraction of an NTA with broadcast synchronization satisfies a safety property, then the original NTA also satisfies it. 
We apply our compositional abstraction framework to the case studies of a producer-consumer system and the clock synchronization protocol in~\cite{new_case_try}, improving verification efficiency compared to the traditional monolithic method.


The rest of the paper is organized as follows. 
Section~\ref{sec: Preliminaries} introduces necessary background knowledge. 
Section~\ref{sec: Timed Transition Systems with Broadcast Actions} introduces the TTSB and corresponding operations of parallel composition and restriction. 
Section~\ref{sec: Two definitions of NTA semantics} introduces the non-compositional and compositional semantics of NTA with broadcast synchronization and proves their equivalence.
Section~\ref{sec: Compositional Abstraction} proposes the timed step simulation for TTSBs and demonstrates the compositionality of the resulting preorder. 
In Section~\ref{sec: Compositional Verification}, we summarize the correctness of our framework and conduct the case studies.
Finally, we discuss the conclusions in Section~\ref{sec: conclusion}. 

\section{Preliminaries} \label{sec: Preliminaries}
We use  $\mathbb{N}$ to denote the set of natural numbers, $\mathbb{R}_{\geq 0}$ the set of non-negative reals, and let  $\mathbb{B}=\{1,0\}$, where $1$ stands for true and $0$ stands for false.

\subsection{Notations for Functions}

The domain of function $f$ is represented as $dom(f)$. If $X$ is a set, then $f\lceil X$ denotes the restriction of $f$ to $X$, forming function $g$ with $dom(g) = dom(f) \cap X$ and $g(z) = f(z)$ for each $z \in dom(g)$. 
The \emph{override} operators~\cite{OverRide} on functions are $\rhd$ and $\lhd$. 
For arbitrary functions $f$ and $g$, $f\rhd g$ denotes the function with $dom(f\rhd g) = dom(f) \cup dom(g)$ such that for all $z \in dom(f \rhd g)$,
\begin{equation*}
    (f\rhd g)(z) \triangleq
    \begin{cases}
        f(z) & \text{if}\ z\in dom(f) \\
	  g(z) & \text{if}\ z\in dom(g)-dom(f)
	\end{cases}
\end{equation*}
and $f\lhd g \triangleq g\rhd f$. Functions $f$ and $g$ are compatible, denoted as $f\heartsuit g$, if $f(z)=g(z)$ for all $z \in dom(f)\cap dom(g)$. For compatible functions $f$ and $g$, their merge is $f\|g \triangleq f \rhd g$. Clearly, $\|$ and $\heartsuit$ are commutative and associative. When we use $f\|g$, it is implicit that $f \heartsuit g$. The notation $f[g]$ represents the \emph{update} of function $f$ according to $g$, defined as $f[g]\triangleq (f\lhd g)\lceil dom(f)$. 

Below are some fundamental properties of functions necessary for the subsequent content of this paper. 

\begin{lemma}\label{lemma: expressions}
For any functions $f$, $g$, and $h$, and set $X$ the following formulas always hold:
    \begin{equation}
        f \heartsuit g[f]
    \end{equation}
    \begin{equation}
        f \heartsuit g \wedge (f\|g) \heartsuit h \Leftrightarrow f \heartsuit g \wedge f \heartsuit h \wedge g \heartsuit h
    \end{equation}
    \begin{equation}
        f\rhd g = f\| g[f]
    \end{equation}
    \begin{equation}
        f[g][h] = f[h\rhd g]
    \end{equation}
    \begin{equation}
        (f\rhd g)[h] = f[h]\rhd g[h]
    \end{equation}
    \begin{equation}
        f\heartsuit g \Rightarrow f\lceil X \heartsuit g
    \end{equation}
    \begin{equation}
        f\heartsuit g, f\heartsuit h \Rightarrow f \heartsuit (g\rhd h)
    \end{equation}
\end{lemma}

\begin{proof}
    Among the formulas above, (1)$\sim$(5) are proved in \cite{Compositional_Abstraction} and the proofs are straightforward from the definitions. The formulas (6) and (7) are newly introduced. Their proof is provided below.
    \begin{enumerate}[label=(\arabic*), start=6]
        \item 
        Since $dom(f\lceil X)\subseteq dom(f)$ and $f\heartsuit g$, for any $z \in dom(f\lceil X) \cap dom(g)$, $(f\lceil X) (z)=g(z)$, that is, $f\lceil(X) \heartsuit g$.

        \item 
        For any $z \in dom(f)\cap dom(g)$, since $f\heartsuit g$, $f(z) = g(z)$. For any $z\in dom(f)\cap (dom(h) - dom(g))$, since $f\heartsuit h$, $f(z) = h(z)$. Hence, we have $f \heartsuit (g\rhd h)$.
        $\hfill\square$
    \end{enumerate}
\end{proof}



\subsection{Labeled Transition Systems} 
We consider two types of \emph{channels}, i.e., \emph{broadcast channels} and \emph{binary channels}. The former allows non-blocking one-to-many synchronization while the latter is used for \emph{binary synchronization} where one side sends, and the other receives. We use $\Delta$ and $\mathcal{C}$ to represent their respective sets. 
The set of \emph{broadcast actions} is $\mathcal{E}_\Delta\triangleq\{\delta!,\delta?\mid \delta\in \Delta\}$ and the set of \emph{binary actions} is  $\mathcal{E}_{\mathcal{C}}\triangleq\{c!,c?\mid c\in \mathcal{C}\}$. The action marked with $!$ or $?$ is called \emph{output action} or \emph{input action}, respectively.
We assume that there is a special \emph{internal action} represented as $\tau$ and \emph{time-passage actions} represented as non-negative real numbers in $\mathbb{R}_{\geq 0}$. We consider \emph{labeled transition systems} associated with the action set $Act \triangleq \mathcal{E}_{\Delta}\cup\mathcal{E}_{\mathcal{C}}\cup\{\tau\}\cup\mathbb{R}_{\geq 0}$. 

\begin{definition}[LTS] \label{def: LTS}
    {\rm{A labeled transition system (LTS)}} is a tuple 
    \begin{equation}
        \mathcal{L}=\langle S,s^0,Act,\rightarrow \rangle,
        \nonumber
    \end{equation}
    where $S$ is a set of states, $s^0\in S$ is the initial state, $Act$ is the action set, and $\rightarrow \:\subseteq S \times Act \times S$ is the transition relation. We use $r,s,t,\ldots$ to range over $S$, and write $s\xrightarrow{a}t$ if $(s,a,t)\in\:\rightarrow$. Here, $s$ is the transition source, and $t$ is the target. An $a$-transition is enabled in $s$, denoted as $s \xrightarrow{a}$, if a state $t$ exists such that $s \xrightarrow{a} t$. A state $s$ is reachable iff there exists a sequence of states $s_1,\ldots, s_n$ where $s_1 = s^0$, $s_n = s$ and for all $i<n$ there exists an action $a$ such that $s_i\xrightarrow{a}s_{i+1}$.
\end{definition}

\subsection{Networks of Timed Automata}
Let $\mathcal{V}$ be a universal set of typed \emph{variables}, with a subset $\mathcal{X}\subseteq\mathcal{V}$ of \emph{clocks} having domain $\mathbb{R}_{\geq 0}$. A \emph{valuation} for a set $V\subseteq \mathcal{V}$ is a function that maps each variable in $V$ to an element in its domain. We write $\{y_i\mapsto z_i,\dots,y_n\mapsto z_n\}$ for the valuation that assigns value $z_i$ to variable $y_i$, for $i=1,...,n$. We use $V\!al(V)$ to denote the valuations set for $V$. For valuation $v\in V\!al(V)$ and time-passage action $d \in \mathbb{R}_{\geq 0}$, $v\oplus d$ is the valuation for $V$ that increases the clocks by $d$ and leaves the other variables unchanged, that is, for all $y\in V$,

\begin{equation}
    (v\oplus d)(y) \triangleq
    \begin{cases}
        v(y)+d & \text{if}\ \  y\in\mathcal{X} \\
	  v(y) & \text{otherwise}
	\end{cases}
    \nonumber
\end{equation}

A \emph{property} $P$ over $V$ is a subset of $V\!al(V)$. Given $W \supseteq V$ and $v\in V\!al(W)$, we say that $P$ holds in $v$, denoted as $v \models P$, if $v\lceil V \in P$. A property $P$ over $V$ is \emph{left-closed} w.r.t for all $v\in V\!al(V)$ and $d\in \mathbb{R}_{\geq 0}$, $v\oplus d\models P \Rightarrow v \models P$ holds. A property $P$ over $V$ is said \emph{not depend on} a set of variables $W \subseteq V$ if for every $v \in V\!al(V)$ and $u \in V\!al(W)$, $v\models P$ holds iff $v[u]\models P$ holds. 

A network of timed automata is a finite set of timed automata \emph{compatible} with each other and communicating through broadcast and binary channels and shared \emph{external variables}. The state variables of a TA are divided into external and \emph{internal variables}. 
Internal variables are private to the TA and cannot be accessed by others. In contrast, external variables are shared among multiple TAs and can be read and updated by them, enabling communication and coordination. In U{\scriptsize PPAAL}, external variables are defined in global declarations, while internal variables are declared locally within a template.

\begin{definition}[TA]
    A {\rm{timed automaton}} is a tuple $\mathcal{A} = \langle L,K,l^0,E,H,v^0,$ $I,\rightarrow,\rightarrow^u\rangle$, where $L$ represents the set of locations, $K \subseteq L$ denotes the set of committed locations, $l^0 \in L$ is the initial location, $E$ and $H$ are disjoint sets of external and internal variables, respectively. $V = E \cup H$, $v^0 \in \text{Val}(V)$ signifies the initial valuation, and $I: L \rightarrow 2^{\text{Val}(V)}$ assigns a left-closed invariant property to each location, ensuring that $v^0 \models I(l^0)$, 
    \begin{equation*}
        \rightarrow\: \subseteq L\times 2^{V\!al(V)} \times \mathcal{E}_{\Delta}\cup\mathcal{E}_{\mathcal{C}}\cup\{\tau\} \times ({V\!al(V)}\rightarrow {V\!al(V)}) \times L
    \end{equation*}
    is the set of transitions, and $\rightarrow^u\subseteq\rightarrow$ is the set of urgent transitions. We write $l\xrightarrow{g,a,\rho}l'$ if $(l,g,a,\rho,l')\in\rightarrow$, where $l$ and $l'$ are the source and the target, $a$ is the action, $g$ is the guard, and $\rho$ is the update function. A guard must be a conjunction of simple conditions on clocks, differences between clocks, and boolean expressions that do not involve clocks.
\end{definition}

Recall that a property $P$ is left-closed if, for all $v \in \mathit{Val}(V)$ and $d \in \mathbb{R}_{\geq 0}$, the implication ${v \oplus d} \models P \Rightarrow v \models P$ holds. This means that lower bounds on clocks, such as $x \geq 5$ for $x \in \mathcal{X}$, are disallowed in location invariants, as required by U{\scriptsize PPAAL}.
Our restrictions on transition guards are consistent with those of U{\scriptsize PPAAL}. For example, guards such as $x - y < 5 \wedge x > 3$ or $n \neq 1$ are allowed, while expressions like $x > 8 \vee x \leq 1$ or $x = n$ are not permitted, where $x, y \in V \cap \mathcal{X}$ and $n \in V - \mathcal{X}$.
Notably, compared to the TA considered in \cite{Compositional_Abstraction}, the TA considered here additionally includes broadcast actions in $\mathcal{E}_{\Delta}$. Throughout this paper, we employ indices to denote individual system components when dealing with multiple indexed systems. For instance, $H_i$ represents the internal variable set of TA $\mathcal{A}_i$. 

\begin{definition}[NTA]
    Two timed automata $\mathcal{A}_1$ and $\mathcal{A}_2$ are compatible if $H_1\cap V_2 = H_2\cap V_1 =\emptyset$ and $v_1^0\heartsuit v_2^0$. A {\rm{network of timed automata (NTA)}} consists of a finite 
    sequence $\mathcal{N}=\langle\mathcal{A}_1,\dots,\mathcal{A}_n\rangle$ of pairwise compatible timed automata.
\end{definition}

\section{Timed Transition Systems with Broadcast Actions} \label{sec: Timed Transition Systems with Broadcast Actions}

The timed transition systems considered in~\cite{Compositional_Abstraction} support shared variables, binary actions, and committed locations but exclude broadcast actions. To perform compositional abstraction on the NTAs with broadcast channels, in this section, we introduce the timed transition systems with broadcast actions and corresponding operations of parallel composition and restriction.

\subsection{Definition of TTSB} \label{sec: Definition of TTSB}
TTSBs extend LTSs with state variables, transition commitments, and time-related behaviors. We follow the approach in \cite{Compositional_Abstraction}, treating committedness as an attribute of transitions rather than an attribute of locations as in U{\scriptsize PPAAL} to obtain compositional semantics. 
Therefore, when interpreting the semantics of TAs using TTSB, transitions starting from committed locations in the TA are interpreted as committed transitions in the corresponding TTSB.
Obviously, committed transitions have higher priority over uncommitted transitions. 

\begin{definition}[TTSB] \label{def: TTSB}
    A {\rm{timed transition system with broadcast actions}} is a tuple
    \begin{equation}
        \mathcal{T}=\langle E,H,S,s^0,Act,\rightarrow^1,\rightarrow^0 \rangle
    \nonumber
\end{equation} 
where $E,H \subseteq \mathcal{V}$ are disjoint sets of external and internal variables, respectively. $S\subseteq V\!al(V)$ is the set of states, where $V=E\cup H$, and $s^0$ is the initial state. $Act$ is the action set which includes broadcast actions. $\rightarrow^1$,$ \rightarrow^0$ are disjoint sets of committed and uncommitted transitions, respectively. A transition $(s,a,t)\in \rightarrow^b$ can also be denoted as $s\xrightarrow{a,b}t$, where $b\in \mathbb{B}$. 
A state $s$ is considered as a committed state, denoted as $Comm(s)$, iff there is at least one committed transition starting from it, i.e., $s\xrightarrow{a,1}$ for some $a\in Act$.
The underlying LTS of $\mathcal{T}$ is $\langle S,s^0,Act,\rightarrow^1\cup\rightarrow^0\rangle$, denoted as $\mathsf{LTS}(\mathcal{T})$. 

We require the following axioms to hold, for all $s,t\in S$, $a,a'\in Act$, $\sigma \in \mathcal{C} \cup \Delta$, $\delta \in \Delta$, $b\in \mathbb{B}$, $d\in\mathbb{R}_{\geq 0}$ and $u\in V\!al(E)$,
    \begin{equation}
        \label{Axiom:I}
        s\xrightarrow{a,1}\wedge \:s\xrightarrow{a',b} \quad \Rightarrow \quad  a'\in\mathcal{E}_{\mathcal{C}}\cup\mathcal{E}_{\Delta}\vee(a'=\tau\wedge b) 
        \tag{Axiom~I}
    \end{equation}
    \vspace{-3mm}
    \begin{equation}
        \label{Axiom:II}
        s[u]\in S 
        \tag{Axiom~II}
    \end{equation}
    \vspace{-3mm}
    \begin{equation}
        \label{Axiom:III}
        s\xrightarrow{\sigma?,b} \quad \Rightarrow \quad s[u]\xrightarrow{\sigma?,b}
        \tag{Axiom~III}
    \end{equation}
    \vspace{-3mm}
    \begin{equation}
        \label{Axiom:IV}
        s\xrightarrow{d,0} t \quad \Rightarrow \quad t = s \oplus d 
    \tag{Axiom~IV}
    \end{equation}
    \vspace{-3mm}
    \begin{equation}
        \label{Axiom:V}
        s\xrightarrow{\delta?,b}
        \tag{Axiom~V}
    \end{equation}
    \vspace{-3mm}
    \begin{equation}
        \label{Axiom:VI}
        s\xrightarrow{\delta?,b}t \quad \Rightarrow \quad {s\lceil E} = {t\lceil E} 
        \tag{Axiom~VI}
    \end{equation}
\end{definition}

Note that in a TTSB, from a committed state, the outgoing transitions might be uncommitted. 
Axiom~I stipulates that neither time passage nor uncommitted $\tau$-transitions can occur in a committed state. In contrast, uncommitted transitions labeled with broadcast or binary actions can occur in this state since they may synchronize with committed transitions of other TTSBs.
Axiom~II asserts that by updating the values of external variables of a state, the result is still a state. 
Axiom~III affirms that updating external variables does not affect the enabledness of transitions labeled with input actions, whether binary or broadcast, which is crucial for compositionality.
Axiom~IV asserts that if time advances with $d$ units, all the clocks advance by $d$, while other variables remain unchanged.

Axiom~V stipulates that for any broadcast channel $\delta$, all the states in TTSB have the corresponding outgoing $\delta?$-transition, that is, TTSB is \emph{input-enabled} for broadcast actions. This axiom aligns with the constraint in \emph{broadcast protocols} definition~\cite{Constraint-Based_Analysis_of_Broadcast_Protocols,verification_of_broadcast_protocols,Fisman_Izsak_Jacobs_2024} and fundamental assumptions of input actions for each state in TIOA work~\cite{TIOA_Kim,I/O_automata_1989}. 
As elaborated in Section~\ref{sec: Parallel Composition} about the parallel composition of TTSBs, this axiom reduces broadcast synchronization to two scenarios, 
enabling a concise design of our parallel composition operator that also guarantees the non-blocking nature of U{\scriptsize PPAAL} broadcast synchronization.
Notably, this axiom does not restrict TTSB to interpreting a limited TA that can execute $\delta?$-transition in any state, i.e., there is no requirement for each location of the TA to have an outgoing $\delta?$-transition, which avoids the cumbersome construction of such TA from a general one. 
Later in Section~\ref{sec: Compositional Semantics} focusing on the compositional semantics of NTAs, we will define the TTSB semantics for a general TA. In terms of the designed rule that introduces suitable self-loop transitions in the TTSB associated with a TA, the semantic conforms with this axiom without any preprocessing of the TA model. Correctness of the rule design is guaranteed by the equivalence between the non-compositional semantics and the compositional semantics of an NTA, which is also proved in Section~\ref{sec: Compositional Semantics}.


Axiom VI is introduced to address the problem caused by a kind of transitions involving broadcast actions and updates of shared variables.
According to the U{\scriptsize PPAAL} help menu, in broadcast synchronization, the update on the $\delta!$-transition is executed first, then those on the $\delta?$-transitions are executed left-to-right in the order of the TAs given in the system definition. This means that the order of the components affects the final composition, as illustrated in Example~\ref{exp: system order}.

\begin{example} \label{exp: system order}
    
    As shown in Fig.~\ref{Updates in broadcast synchronization}, when defining the system, if TA $\mathcal{A}_2$ is to the left of $\mathcal{A}_3$, i.e., $\mathcal{N} = \langle \mathcal{A}_1, \mathcal{A}_2, \mathcal{A}_3 \rangle$, then after executing the broadcast synchronization via $\delta$, the value of the external variable $n$ is $2$. In contrast, if $\mathcal{N}' = \langle \mathcal{A}_1, \mathcal{A}_3, \mathcal{A}_2 \rangle$, the value of $n$ becomes $1$.
\end{example}

\begin{figure}
    \vspace{-3mm}
    \centering
    \includegraphics[height=1.1cm]
    {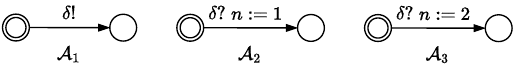}
    \vspace{-1mm}
    \caption{Updates in broadcast synchronization} 
    \label{Updates in broadcast synchronization}
    \vspace{-4mm}
\end{figure}

Therefore, to make our framework compatible with U{\scriptsize PPAAL} semantics while avoiding the occurrence of this scenario, we allow value updates of variables in $\delta!$-transition but introduce Axiom~VI to forbid the value updates of the external variables in $\delta?$-transition.
Although this results in some loss of modeling capability for TTSB, ensuring the associativity of TTSB's parallel composition rules introduced in Section~\ref{sec: Parallel Composition} is crucial.

\subsection{Parallel Composition} \label{sec: Parallel Composition}
We now introduce the parallel composition operation on TTSBs. It is a partial operation, defined only when TTSBs are compatible: the internal variable set of one TTSB must not overlap with the variable set of the other, and their initial states must be compatible.
Recall that $E_i$ (resp. $H_i$) represents the external (resp. internal) variable set of TTSB $\mathcal{T}_i$, $V_i = E_i \cup H_i$, $\Delta$ (resp. $\mathcal{C}$) is the broadcast (resp. binary) channel set, and $\mathcal{E}_{\mathcal{C}}$ is the set of binary actions.

\begin{definition}[Parallel composition] \label{def: Parallel composition}
    Two TTSBs $\mathcal{T}_1$ and $\mathcal{T}_2$ are compatible if $H_1 \cap V_2 = H_2 \cap V_1 = \emptyset$ and $s^0_1 \heartsuit s^0_2$. Their parallel composition $\mathcal{T}_1 \| \mathcal{T}_2$ is denoted as $\mathcal{T} = \langle E, H, S, s^0, Act, \rightarrow^1, \rightarrow^0 \rangle$, where $E = E_1 \cup E_2$, $H = H_1 \cup H_2$, $S = \{r \| s \mid r \in S_1 \wedge s \in S_2 \wedge r \heartsuit s\}$, $s^0 = s^0_1 \| s^0_2$, and $\rightarrow^1$, $\rightarrow^0$ are the least relations satisfying the rules in Fig.\ref{fig:Trans_Build_Rules}. Here, $i,j \in \{1,2\}$, $r, r'\in S_i$, $s,s'\in S_j$, $b, b'\in\mathbb{B}$, $\delta\in\Delta$, $c\in\mathcal{C}$, $a\in\mathcal{E}_{\mathcal{C}}$ and $d\in\mathbb{R}_{\ge 0}$.
\end{definition}

\begin{figure}[h!]
    \vspace{-5mm}
    \centering
    \begin{tikzpicture}
        \draw (6.4,3) rectangle (12.2,4.22);
        \draw (0,3) rectangle (6.4,4.22);
        \draw (6.4,1.3) rectangle (12.2,3);
        \draw (0,1.3) rectangle (6.4,3);
        \draw (6.4,0) rectangle (12.2,1.3);
        \draw (0,0) rectangle (6.4,1.3);

        \node at (6.4/2,3+1.22/2) 
        {$\inferenceOne{EXT}
        {r\|s \xrightarrow{a,b}r'\rhd s}
        {r\xrightarrow{a,b}_i r'}
        $};
        
        \node at (6.4+5.8/2,3+1.22/2)
        {$\inferenceOne{TAU}
        {r\|s \xrightarrow{\tau,b}r'\rhd s}
        {r\xrightarrow{\tau,b}_i r' \myHspace
        Comm(s)\Rightarrow b}$};

        \node at (6.4/2,1.3+1.74/2) 
        {$\inferenceOne{SYNC}
        {r\|s \xrightarrow{\tau,b\vee b'} r'\lhd s'}
        {\begin{array}{c}
            {r\xrightarrow{c!,b}_i r' \myHspace
            {s[r']\xrightarrow{c?,b'}_j s'} \myHspace
            i\not = j} \\
            {Comm(r)\vee Comm(s)\Rightarrow b\vee b'}
        \end{array}}
        $};
        
        \node at (6.4+5.8/2,1.3+1.74/2) 
        {$\inferenceOne{TIME}
        {r\|s \xrightarrow{d,0} r'\|s'}
        {r\xrightarrow{d,0}_i r' \myHspace 
        s\xrightarrow{d,0}_j s' \myHspace
        i\not = j}
        $};

        \node at (6.4/2,1.32/2)
        {$\inferenceOne{SND}
        {r\|s \xrightarrow{\delta!,b\vee b'} r'\| s'}
        {\begin{array}{c}
            r\xrightarrow{\delta!,b}_i r' 
            \myHspace {s[r']\xrightarrow{\delta?,b'}_j s'} 
            \myHspace i\not = j 
        \end{array}}$};

        \node at (6.4+5.8/2,1.32/2) 
        {$\inferenceOne{RCV}
        {r\|s \xrightarrow{\delta?,b\vee b'} r'\| s'}
        {\begin{array}{c}
            r\xrightarrow{\delta?,b}_i r' \myHspace
            {s\xrightarrow{\delta?,b'}_j s'} \myHspace i\not = j 
        \end{array}}$};

    \end{tikzpicture}
    \vspace{-2mm}
    \caption{Rules for parallel composition of TTSBs}
    \label{fig:Trans_Build_Rules}
    \vspace{-3mm}
\end{figure}

Since broadcast transitions do not interact with binary transitions, internal transitions, or time passage, the parallel composition rules for these transitions are consistent with those in TTSs, as shown in rules \textbf{EXT}, \textbf{TAU}, \textbf{SYNC}, and \textbf{TIME} in Fig.~\ref{fig:Trans_Build_Rules}.
Rule \textbf{EXT} specifies that a transition labeled with binary action $a$ in component $\mathcal{T}_i$ from state $r$ leads to a corresponding transition in the composition $\mathcal{T}$. Occurrence of the $a$-transition may override some of the shared variables, and $r'\rhd s$ is again a state of $\mathcal{T}$. Rule \textbf{TAU} states that a $\tau$-transition from $r$ in $\mathcal{T}_i$ induces a corresponding transition in $\mathcal{T}$, except for the case where the $\tau$-transition is uncommitted and $\mathcal{T}_j$ is in a committed state. Rule \textbf{SYNC} describes the binary synchronization between components. If $\mathcal{T}_i$ has a $c!$-transition from $r$ to $r'$, and $\mathcal{T}_j$ has a corresponding $c?$-transition from $s[r']$, i.e. state $s$ updated by $r'$, to $s'$, then the composition will have a $\tau$-transition to from $r\|s$ to $r'\lhd s'$. We say a binary synchronization is committed if at least one involved transition is committed, that is, $b\vee b' = 1$. The condition $Comm(r)\vee Comm(s) \Rightarrow b \vee b'$ implies that a committed binary synchronization can always occur, and an uncommitted binary synchronization can only occur when neither of the components is in a committed state. 
Rule \textbf{TIME} states that time progresses at the same rate in both components. 

For the composition of broadcast transitions, generally, there are four scenarios, $(\delta!,\delta?)$, $(\delta?,\delta?)$, $(\delta!,\cdot)$ and $(\delta?,\cdot)$ in two components. By Axiom~V that each TTSB is input-enabled for broadcast actions, we only need to tackle the composition for the first two scenarios.
We first consider the designed rule \textbf{SND} for $(\delta!,\delta?)$ transitions. Different from rule \textbf{SYNC} that generates a $\tau$-transition in $\mathcal{T}$, rule \textbf{SND} states that if $\mathcal{T}_i$ has a $\delta!$-transition from $r$ to $r'$, and $\mathcal{T}_j$ has $\delta?$-transition from $s[r']$ to $s'$, the composition will have a transition from $r\|s$ to $r'\|s'$, which is still labeled with $\delta!$. Intuitively, this rule allows a $\delta!$-transition, after synchronizing with a $\delta?$-transition, to synchronize with $\delta?$-transitions in other components.
By Axiom~VI, since the shared variable is not updated in the $\delta?$-transition, we have $r'\|s' = r'\lhd s'$.
As the compositional framework also addresses scenarios combining broadcast synchronization and committed locations in NTAs, we must consider committedness for the rule \textbf{SND} in two aspects.
For the first one, whether the composed $\delta!$-transition is committed is 
decided by the value of $b\vee b'$.  For the second aspect, it should be noted that unlike rule \textbf{SYNC}, rule \textbf{SND} does not have the condition $Comm(r)\vee Comm(s) \Rightarrow b \vee b'$.  
This is because having this condition would cause the associativity violation of the parallel composition operation of TTSBs,
which is illustrated by Example~\ref{exp: committedness}. 

\begin{example}\label{exp: committedness}
Consider the three TTSBs shown in Fig.\ref{Composition of three TTSBs}, where $r$ and $t$ are committed states. 
Suppose we add the condition $Comm(r)\vee Comm(s) \Rightarrow b \vee b'$ to rule \textbf{SND}. If we compose $\mathcal{T}_1$ and $\mathcal{T}_2$ first, since $r\xrightarrow{\delta?,0}r$ and $s\xrightarrow{\delta!,0}s'$ are both uncommitted, then $Comm(r)\vee Comm(s) \Rightarrow b \vee b'$ values false. So $\mathcal{T}_1\|\mathcal{T}_2$ does not have $\delta!$-transitions, resulting in the absence of $\delta!$-transition in the composition of $\mathcal{T}_1\|\mathcal{T}_2$ and $\mathcal{T}_3$, i.e. $\mathcal{T}_1\|\mathcal{T}_2\|\mathcal{T}_3$. 
However, if we compose $\mathcal{T}_2$ and $\mathcal{T}_3$ first, 
the final composition $\mathcal{T}_1\|(\mathcal{T}_2\|\mathcal{T}_3)$ will contain a committed $\delta!$-transition. 
Obviously, this violates the associative requirement.
\end{example}

\begin{figure}
    \vspace{-3mm}
    \centering
    \includegraphics[height=1.2cm]
    {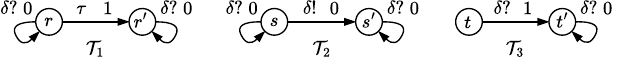}
    \vspace{-1mm}
    \caption{Composition of three TTSBs} 
    \label{Composition of three TTSBs}
    \vspace{-3mm}
\end{figure}

In contrast, our current design of rule \textbf{SND} ensures the parallel composition operator associative, which will be shown by Theorem~\ref{theroem: Commutativity and Associativity} at the end of this section.
We now prove that the target states of the transitions generated by rule \textbf{SND} are always states of $\mathcal{T}$. By Axiom~II for $\mathcal{T}_j$, it follows that $s[r']$ is a state of $\mathcal{T}_j$. Further by Lemma~\ref{lemma: expressions}(1), we have $r'\heartsuit s[r']$, then by Lemma~\ref{lemma: expressions}(3), $r'\heartsuit s[r']\lceil E_j$. By Axiom~VI, $s[r']\lceil E_j = s'\lceil E_j$, which means $r' \heartsuit s'\lceil E_j$. Since $V_i\cap H_j=\emptyset$, $r'\heartsuit s'\lceil H_j$ holds. Finally by Lemma~\ref{lemma: expressions}(7), we obtain $r'\heartsuit s'$. Hence, $r'\|s'$ is a state of $\mathcal{T}$.


We now consider the designed rule \textbf{RCV} for $(\delta?,\delta?)$ transitions, which states that if $\mathcal{T}_i$ has a $\delta?$-transition from $r$ to $r'$ and $\mathcal{T}_j$ has a $\delta?$-transition from $s$ to $s'$, the composition will have a $\delta?$-transition from $r\|s$ to $r'\|s'$.
Like the composed $\delta!$-transition in rule \textbf{SND}, the composed $\delta?$-transition in rule \textbf{RCV} has the committedness $b\vee b'$. 
Still, to guarantee associativity, rule \textbf{RCV} does not have the condition $Comm(r)\vee Comm(s) \Rightarrow b \vee b'$. 
The target state of the generated $\delta?$-transition also remains a state of $\mathcal{T}$. Since $r\| s$ is a state of $\mathcal{T}$, we have $r\heartsuit s$, which implies $r\lceil E_i \heartsuit s\lceil E_j$ by Lemma~\ref{lemma: expressions}(6). Further by Axiom~VI, neither $r\xrightarrow{\delta?,b}_i r'$ nor $ s\xrightarrow{\delta?,b'}_j s'$ modifies the value of external variables, i.e. $r\lceil E_i = r'\lceil E_i, s\lceil E_j = s'\lceil E_j$. Based on this and $r\lceil E_i \heartsuit s\lceil E_j$, we have $r'\lceil E_i \heartsuit s'\lceil E_j$. Finally, since $H_i\cap V_j = H_j\cap V_i = \emptyset$, by Lemma~\ref{lemma: expressions}(7), we get $r'\heartsuit s'$, following that $r'\|s' \in S$.


The parallel composition operation on TTSBs is well-defined, that is, the composition of two TTSBs remains a TTSB.
The proof is in Appendix~\ref{APP: proof of Composition well-defined}.

\begin{lemma} [Composition well-defined] \label{lemma: Composition well-defined}
    Let $\mathcal{T}_1$ and $\mathcal{T}_2$ be compatible TTSBs. Then $\mathcal{T}_1\|\mathcal{T}_2$ is a TTSB.
\end{lemma}

Notably, the parallel composition operator defined in this paper satisfies two crucial properties for compositional abstraction: commutativity and associativity. This is a main theorem in this paper, 
and the proof is in Appendix~\ref{APP: proof of Commutativity and Associativity}.


\begin{theorem}[Commutativity and Associativity] 
\label{theroem: Commutativity and Associativity}
    The parallel composition operation on TTSBs is commutative and associative.
\end{theorem}


\subsection{Restriction} \label{sec: Restriction}

The designed parallel composition rules allow three types of component communication through broadcast channels, binary channels, and shared variables. For the first two types, when no matching component is available, a broadcast or binary transition can be removed or replaced with a $\tau$-transition. For the third type, if an external variable is no longer used in other components, it can be converted to an internal one. 
Here, we introduce the restriction operation to handle these channels and variables. This operation not only enables simpler abstractions but is also crucial for establishing the correct compositional semantics of NTAs in Section~\ref{sec: Compositional Semantics}.





\begin{definition}[Restriction for TTSB] \label{def: Restriction for TTSB}
    Given a TTSB $\mathcal{T}$ and a set $C \subseteq \Delta\cup\mathcal{C}\cup E$ of broadcast, binary channels, and external variables, we denote the $\mathcal{T}$ restricted by $C$ as $\mathcal{T}\backslash C$. The TTSB $\mathcal{T}\backslash C$ is identical to $\mathcal{T}$, except that for any transition $s\xrightarrow{a,b}s'$ of $\mathcal{T}$:
    \begin{enumerate}
        \item 
        If $a \in \{\delta?,c!,c?\mid \delta \in C \cap \Delta, c\in C\cap \mathcal{C}\}$, it will be removed from $\mathcal{T}\backslash C$.

        \item 
        If $a \in \{\delta!\mid \delta \in C \cap \Delta\}$ and $Comm(s)\wedge \neg b$, it will be removed from $\mathcal{T}\backslash C$.

        \item
        If $a \in \{\delta!\mid \delta \in C \cap \Delta\}$ and $Comm(s)\Rightarrow b$, it will be replaced by $s\xrightarrow{\tau,b}s'$ in $\mathcal{T}\backslash C$.  
    \end{enumerate}
    \vspace{-1mm}
    and the external and internal variable sets of $\mathcal{T}\backslash C$ are $E-C$ and $H\cup (E\cap C)$, respectively.
\end{definition}

The first type of transitions labeled with input broadcast actions or binary actions is removed because they cannot occur independently. For the output broadcast transitions, i.e., $\delta!$-transitions, due to the non-blocking nature of broadcast synchronization, they can occur independently but with consideration of their committedness and the committedness of their source states. Recall that a committed state can have uncommitted outgoing transitions. So, a $\delta!$-transition can be uncommitted while its source state is committed. This also implies that the source state has another outgoing committed transition. In this case, this $\delta!$-transition should be removed because it cannot occur due to its low priority. Otherwise, for the cases where the sourcing state is not committed or the $\delta!$-transition itself is committed, the $\delta!$-transition can occur and should be replaced with a $\tau$-transition, since no other transitions can synchronize with it. In summary, the second and third types of transitions are handled based on both their own committedness and that of their source states. Here, committedness can be seen as a binary priority. We plan to extend our framework to support more general priority relations in the future, which is a meaningful enhancement.


Obviously, $\mathsf{LTS}(\mathcal{T})=\mathsf{LTS}(\mathcal{T}\backslash C)$, if $C\subseteq E$. We write $\Sigma(\mathcal{T})$ for the set of channels that are enabled in the transitions of $\mathcal{T}$. Using this notation, we can formulate some restriction laws, such as $\mathcal{T}\backslash C = \mathcal{T}$ if $\Sigma(\mathcal{T})\cap C=E\cap C=\emptyset$.


\section{Two definitions of NTA semantics}\label{sec: Two definitions of NTA semantics}
This section introduces two definitions of the semantics of NTA with broadcast channels. 
One strictly follows U{\scriptsize PPAAL} semantics by constructing an LTS directly, but is not compositional, therefore called non-compositional semantics in this paper.
The other is compositional, which is achieved by associating TTSBs to each TA, applying parallel composition and restriction operations, and finally extracting the underlying LTS. We prove that these two semantics are equivalent, which is also a main theorem to implement compositional abstraction for
timed systems with broadcast synchronization.


For the compositional semantics of NTA, we first impose some axioms that U{\scriptsize PPAAL} does not require on timed automata to obtain compositionality. For any TA $\mathcal{A}=\langle L,K,l^0,E,H,v^0,I,\rightarrow \rangle$, we require:
    \begin{equation*}
        I(l)\ does\ not\ depend\ on\ E \qquad
        \tag{Axiom~VII}
    \end{equation*}
    \vspace{-5mm}
    \begin{equation*}
        l\xrightarrow{g,\sigma?,\rho}l' \quad \Rightarrow \quad g \ does \ not \ depend \ on\ E
        \tag{Axiom~VIII}
    \end{equation*}
    \begin{equation*}
        \forall l\in K\ \forall v\in I(l) \ \exists(l\xrightarrow{g,a,\rho}l'):v\models g\wedge \rho(v)\models I(l')
        \tag{Axiom~IX}
    \end{equation*}
    \begin{equation*}
        l \xrightarrow{g,a,\rho}^u l' \quad \Rightarrow \quad a = \tau \wedge g \ does \ not \ depend \ on \ \mathcal{X}
        \tag{Axiom~X}
    \end{equation*}
    \begin{equation*}
        l \xrightarrow{g,\delta?,\rho} l' \quad \Rightarrow \quad  \rho \ does\ not\ update \ {V\!al(E)}
        \tag{Axiom~XI}
    \end{equation*}

Axioms~VII-X are similar to the corresponding axioms of TTS in \cite{Compositional_Abstraction}, with the new constraints with regard to broadcast synchronization.
Axiom~VII is introduced to avoid runtime errors in the scenario: modification of external variables in one automaton causes the violation of location invariant in another,
which leads the timed system to reach an undefined state in U{\scriptsize PPAAL}.
Axiom~VIII ensures that the update function $\rho$ of $\sigma!$-transition does not affect satisfaction of the guard $g$ in the corresponding $\sigma?$-transitions, where $\sigma$ could be either a broadcast or a binary channel.
As shown in~\cite{Compositional_Abstraction}, real-world models rarely violate this axiom.
Axiom~IX requires that for any committed location $l$, a transition must exist starting from it.
This axiom excludes some ``bad'' models that may lead to deadlock and ensures that when associating a TTSB to a TA, the states corresponding to the committed location $l$ must be committed ones.
Axiom~X says that an urgent transition should be internal and not have clock guards.
This effectively excludes most models with urgent transitions, as totally supporting this feature is currently beyond the scope of this work.
Axiom~XI, newly introduced in this paper, corresponds to Axiom~VI for TTSB. It prohibits the values of external variables from being updated in the update function $\rho$ for a $\delta?$-transition, thus guaranteeing the associativity of parallel composition operation.
Empirically, real-world models violating this axiom are uncommon.
Overall, these axioms impose acceptable restrictions on TAs while ensuring their applicability to modeling most timed systems.


\subsection{Non-compositional Semantics}\label{sec: Non-compositional Semantics}

Strictly following the U{\scriptsize PPAAL} help menu, we give the non-compositional LTS semantics of an NTA, in which all the TAs satisfy Axiom~VII$\sim$XI. 
It can be viewed as a further formalization of the official semantics and is essential for the subsequent proof of semantics equivalence.

\begin{definition}[LTS semantics of NTA]\label{def: LTS semantics of NTA}
    Let $\mathcal{N}=\langle\mathcal{A}_1,\dots,\mathcal{A}_n\rangle$ be an NTA. Let $V=\bigcup_{i=1}^n (V_i\cup{\mathsf{loc}_i})$, where $\mathsf{loc}_i$ is a fresh variable with type $L_i$. The LTS semantics of $\mathcal{N}$, denoted as $\mathsf{LTS}(\mathcal{N})$, is the LTS $\langle S,s^0,\rightarrow \rangle$, where
    \begin{align*}
        S&=\{v\in{V\!al}(V)\mid\forall i : v \models I_i(v(\mathsf{loc}_i))\},
        \\
        s^0&=v_1^0\|\cdots\|v_n^0\|\{\mathsf{loc}_1\mapsto l_1^0,\dots,\mathsf{loc}_n\mapsto l_n^0\},
    \end{align*}
    and $\rightarrow$ is defined by the rules in Fig.\ref{fig:NTA_UPPAAL_Semantics}.  
\end{definition}

\begin{figure}[h!]
    \vspace{-6mm}
    \centering
    \begin{tikzpicture}
        \draw (0,0) rectangle (12,2.5);
        \draw (0,2.5) rectangle (12,3.8);
        \draw (0,3.8) rectangle (12,5.8);
        \draw (0,5.8) rectangle (12,7.4);

        \node at (12/2,5.8+1.7/2)
        {$\inferenceOne{TAU}
        {s\xrightarrow{\tau}s'}
        {\begin{array}{c@{\myHspace}c}
             l\xrightarrow{g,\tau,\rho}_i l' & s(\mathsf{loc}_i)=l \myHspace s'=\rho(s)[\{\mathsf{loc}_i\mapsto l'\}]  \\
             s\models g & (\forall k:s(\mathsf{loc}_k\not \in K_k)\vee l \in K_i 
        \end{array}}$};
        \node at (12/2,3.8+2.1/2)
        {$\inferenceOne{SYNC}
        {s\xrightarrow{\tau}s'}
        {\begin{array}{c@{\myHspace}c@{\myHspace}c}
            l_i\xrightarrow{g_i,c!,\rho_i}l'_i &
            l_j\xrightarrow{g_j,c?,\rho_j}l'_j &
            s'=\rho_j(\rho_i(s))[\{\mathsf{loc}_i\mapsto l'_i,\mathsf{loc}_j\mapsto l'_j\}]\\
            s(\mathsf{loc}_i)=l_i &
            s(\mathsf{loc}_j)=l_j & 
            (\forall q:s(\mathsf{loc}_q)\not \in K_q)\vee l_i\in K_i \vee l_j \in K_j\\
            s\models g_i &
            s\models g_j &
            i\not = j
        \end{array}}$};
        
        \node at (12/2,2.5+1.25/2)
        {$\inferenceOne{TIME}
        {s\xrightarrow{d}s'}
        {s'=s\oplus d \myHspace 
        \forall k:s(\mathsf{loc}_k)\not \in K_k  \myHspace
        \nexists(l\xrightarrow{g,\tau,\rho}_i^u l'):s(\mathsf{loc}_i)=l\vee s\models g}
        $};
        
        \node at (12/2,2.6/2)
        {$\inferenceOne{BCST}
        {s\xrightarrow{\tau}s'}
        {\begin{array}{c}
            l_i\xrightarrow{g_i,\delta!,\rho_i}l'_i \myHspace\myHspace\myHspace\myHspace
            s(\mathsf{loc}_i)=l_i \myHspace\myHspace\myHspace\myHspace
            s\models g_i \\
            \forall j\in RS(\delta,i,s):l_j\xrightarrow{g_j,\delta?,\rho_j}l_j',s(\mathsf{loc}_j)=l_j, s\models g_j \\
            (\forall q:s(\mathsf{loc}_q)\not \in K_q)\vee l_i\in K_i\vee(\exists j\in RS(\delta,i,s): l_j\in K_j)\\
            s'= \rho_{j_m}(\cdots\rho_{j_1}(\rho_i(s)))[\{\mathsf{loc}_i\mapsto l_i', \mathsf{loc}_{j_1}\mapsto l_{j_1}',\dots,\mathsf{loc}_{j_m}\mapsto l_{j_m}'\}]
        \end{array}}$};
        
    \end{tikzpicture}
    \caption{LTS semantics of an NTA in U{\scriptsize PPAAL}}
    \label{fig:NTA_UPPAAL_Semantics}
    \vspace{-4mm}
\end{figure}


Rule \textbf{TAU}, \textbf{SYNC}, and \textbf{TIME} respectively describe the internal transitions of each TA in the NTA, the binary synchronization between TAs, and the passage of time. 
We refer to \cite{Compositional_Abstraction} for the detailed description of these rules.
Rule \textbf{BCST}
describes broadcast synchronization among TAs. According to the U{\scriptsize PPAAL} help menu, when the NTA $\mathcal{N}$ is in a certain state $s$ and a $\delta!$-transition in a certain $\mathcal{A}_i$ is activated, all other TAs in $\mathcal{N}$ with executable $\delta?$-transitions must select one to synchronize with it. This is described in the first two lines of rule \textbf{BCST}, where the set $RS(\delta,i,s)$, defined as the set of indices of all the TAs with executable $\delta?$-transitions\footnote{For the guard $g_j$ in the input broadcast transition $l_j \xrightarrow{g_j, \delta?, \rho_j} l_j'$, although the current U{\scriptsize PPAAL} help menu states that it can not have clock constraints, the U{\scriptsize PPAAL} change log states that it can have clock constraints since version 4.1.3, and our framework also supports clock constraints in broadcast transitions.}, 
is $\{j \mid i \neq j, s(\mathsf{loc}_j) = l_j, \exists(l_j \xrightarrow{g_j, \delta?, \rho_j} l_j'): s \models g_j\}$.
Considering that in the current state $s$, some TAs can be in committed locations, the third line is imposed to clarify that the broadcast synchronization can only occur under two conditions: 
1) no TA is in a committed location, 
2) at least one of the TAs participating in the synchronization is in a committed location.
In the last line, $\rho_i(s)$ is defined as $v[\rho_i(v \lceil V_i)]$, meaning that if an update function $\rho_i : {V\!al}(V_i) \rightarrow {V\!al}(V_i)$ is applied to state $s$, it only modifies the variables in $V_i$.
This line provides the update rule: $\rho_i$ in the $\delta!$-transition is executed first, then $\rho_{j_1}, \dots, \rho_{j_m}$ in the $\delta?$-transitions, where $j_1, \dots, j_m$ are the indices in $RS(\delta,i,s)$. Notably, due to Axiom~XI, arbitrary execution of $\rho_{j_1}, \dots, \rho_{j_m}$ will result in the same target state $s'$, i.e., the result is deterministic.

\subsection{Compositional Semantics}\label{sec: Compositional Semantics}

To derive the compositional semantics of an NTA, we first obtain the TTSB semantics for each TA in the NTA, then compose them into a single TTSB, apply restrictions, and finally extract the underlying LTS.


\begin{definition}[TTSB semantics of TA]\label{def: TTSB semantics of TA}
    Let $\mathcal{A}=\langle L,K,l^0,E,H,v^0,I,\rightarrow \rangle$ be a TA. The TTSB associated to $\mathcal{A}$, denoted as $\mathsf{TTSB}(\mathcal{A})$, is the tuple
    \begin{equation}
        \langle E, H\cup\{\mathsf{loc}\}, S,s^0,\rightarrow^1,\rightarrow^0 \rangle,
        \nonumber
    \end{equation}
    where $\mathsf{loc}$ is a fresh variable with type $L$. Let $W=E\cup H \cup \{\mathsf{loc}\}$, $S=\{v\in {V\!al}(W)\mid v\in I(v(\mathsf{loc}))\}$, $s^0=v^0\|\{\mathsf{loc}\mapsto l^0\}$. The transitions are defined by the rules in Fig.\ref{fig:TA_TTSB_semantics}.
\end{definition}

\begin{figure}[h!]
    \centering
    \begin{tikzpicture}
        \draw (0,2.6) rectangle(12,3.9);
        \draw (0,1.3) rectangle (12,2.6);
        \draw (0,0) rectangle (12,1.3);
        \node at (12/2,2.6+1.35/2) 
        {$\inferenceOne{ACT}
            {s\xrightarrow{a,b}s'}
            {
                l\xrightarrow{g,a,\rho}l' \myHspace
                s(\mathsf{loc})=l \myHspace
                s\models g \myHspace
                s' =  \rho(s)[\{\mathsf{loc}\mapsto l'\}] \myHspace
                b\Leftrightarrow(l\in K)}
        $};

        \node at (12/2,1.3+1.35/2) 
        {$\inferenceOne{TIME}
        {s \xrightarrow{d,0}s'}
        {
            s'=s\oplus d \myHspace
            s(\mathsf{loc})\not\in K \myHspace
            \nexists(l\xrightarrow{g,\tau,\rho}^u l') : s(\mathsf{loc})=l\wedge s \models g}
        $};

        \node at (12/2,1.35/2) 
        {$\inferenceOne{VIRT}
            {s\xrightarrow{\delta?,0}s}
            {
                s(\mathsf{loc})=l \myHspace \exists \: \delta \:\forall \: l \xrightarrow{g,\delta?,\rho}l' : s\not\models g
            }
        $};
        
    \end{tikzpicture}
    \caption{TTSB semantics of a TA}
    \label{fig:TA_TTSB_semantics}
\end{figure}

Rule \textbf{ACT} describes state transitions in $\mathcal{A}$ caused by broadcast, binary, and internal actions, where the action $a\in\mathcal{E}_\Delta\cup\mathcal{E}_{\mathcal{C}}\cup\{\tau\}$. Rule \textbf{TIME} describes delay transitions, where $d\in\mathbb{R}_{\geq 0}$. 
Given a broadcast channel $\delta$, if state $s$ (with $s(\mathsf{loc}) = l$) does not have outgoing $\delta?$-transition, then Rule \textbf{VIRT} will generate an additional self-loop $\delta?$-transition for $s$,
which ensures
the satisfaction of Axiom~V discussed in Section~\ref{sec: Definition of TTSB}.
In what follows, Theorem~\ref{Theorem: semantics equal}, stating the equivalence between the compositional and the non-compositional semantics, implies that the additional transitions will not affect the correctness of final NTA semantics.
Note that the generated self-loop $\delta?$-transition must be non-committed, even when $l$ is a committed location.
This is consistent with the definition of $Comm(s)$ that allows for non-committed outgoing transitions from $s$, further avoids the committedness change of
the transitions generated by the subsequent parallel composition. 
Without this requirement, the additional self-loop $\delta?$-transition is designed to be committed. If there is an uncommitted transition labeled with $\delta!$ or $\delta?$ in another component, then the composed transition will be incorrectly turned into a committed one.
We can prove that the structure obtained from $\mathcal{A}$ by this definition is indeed a TTSB, 
and Appendix~\ref{APP: proof of TTSB(A) is a TTSB} shows the details.

\begin{lemma} \label{lemma: TTSB(A) is a TTSB}
    $\mathsf{TTSB}(\mathcal{A})$ is a TTSB.
\end{lemma}

Based on the TTSB semantics of TA, the LTS semantics of a given NTA $\mathcal{N} = \langle \mathcal{A}_1, \dots, \mathcal{A}_n \rangle$, can be represented by the following expression:

\vspace{-2mm}

\begin{equation*}
    \mathsf{LTS}((\mathsf{TTSB}(\mathcal{A}_1)\|\cdots\|\mathsf{TTSB}(\mathcal{A}_n)\backslash (\Delta \cup \mathcal{C}))
\end{equation*}





The following theorem, which is proven in 
Appendix~\ref{APP: proof of semantics equal}
in detail, states that the compositional semantics of NTAs with broadcast channels, defined in terms of TTSBs, is equivalent (modulo isomorphism) to the non-compositional semantics defined Definition~\ref{def: LTS semantics of NTA}.
This implies that our design of the compositional semantics of NTAs with broadcast channels is correct.

\begin{theorem} \label{Theorem: semantics equal}
    Let $\mathcal{N}=\langle\mathcal{A}_1,\dots,\mathcal{A}_n \rangle$ be an NTA. Then 
    \begin{equation*}
        \mathsf{LTS}(\mathcal{N}) \cong \mathsf{LTS}((\mathsf{TTSB}(\mathcal{A}_1)\|\cdots\|\mathsf{TTSB}(\mathcal{A}_n))\backslash (\Delta \cup \mathcal{C})).
    \end{equation*}
\end{theorem}

\section{Compositional Abstraction}\label{sec: Compositional Abstraction}
This section introduces the timed step simulation for TTSBs. It demonstrates the compositionality of the
induced preorder, providing formal support for the compositional abstraction of timed systems with broadcast synchronization.

\begin{definition}[Timed step simulation for TTSBs]\label{def: Timed step simulation for TTSBs}
    Two TTSBs $\mathcal{T}_1$ and $\mathcal{T}_2$ are comparable if  $E_1 = E_2$. Given comparable TTSBs $\mathcal{T}_1$ and $\mathcal{T}_2$, we say that a relation $\rm{R}\subseteq S_1\times S_2$ is a timed step simulation from $\mathcal{T}_1$ to $\mathcal{T}_2$, provided that $s_1^0\:\rm{R}\:s_2^0$ and if $r\:{\rm{R}}\:s$, then
    \begin{enumerate}
        \item $r\lceil E_1 = s\lceil E_2$,
        \item $\forall u \in {V\!al}(E_1):r[u]\:{\rm{R}}\:s[u]$,
        \item if $Comm(s)$ then $Comm(r)$,
        \item if $r\xrightarrow{a,b}r'$ then either there exists an $s'$ such that $s\xrightarrow{a,b}s'$ and $r'\:{\rm{R}}\:s'$, or $a=\tau$ and $r'\:{\rm{R}}\:s$.
    \end{enumerate}
\end{definition}

We denote $\mathcal{T}_1 \preceq \mathcal{T}_2$ when there exists a timed step simulation from $\mathcal{T}_1$ to $\mathcal{T}_2$. 
If $\mathcal{T}_1 \preceq \mathcal{T}_2$, then $\mathcal{T}_2$ can either simulate the transitions of $\mathcal{T}_1$ or do nothing if the transition is internal. However, $\mathcal{T}_2$ cannot introduce internal transitions that do not exist in $\mathcal{T}_1$. Therefore, the partial order $\preceq$ defined by timed step simulation describes a behavioral relation between timed systems. It requires that $\mathcal{T}_2$ preserves all external behaviors of $\mathcal{T}_1$, but it allows $\mathcal{T}_2$ to omit some internal behaviors. This property is essential for constructing compositional abstractions in Section~\ref{sec: Compositional Verification}.

Based on the definition, it is straightforward to establish that $\preceq$ is reflexive and transitive. Furthermore, we demonstrate that $\preceq$ is a precongruence for parallel composition, which is another main theorem in this paper.
The corresponding proof is in 
Appendix~\ref{APP: proof of precongruence}.

\begin{theorem} \label{theorem: precongruence}
    Let $\mathcal{T}_1$, $\mathcal{T}_2$, $\mathcal{T}_3$ be TTSBs with $\mathcal{T}_1$ and $\mathcal{T}_2$ comparable, $ \mathcal{T}_1 \preceq \mathcal{T}_2$, and both $\mathcal{T}_1$ and $\mathcal{T}_2$ compatible with $\mathcal{T}_3$. Then $\mathcal{T}_1\|\mathcal{T}_3\preceq\mathcal{T}_2\|\mathcal{T}_3$.
\end{theorem}

The timed step simulation preorder $\preceq$ is typically not a precongruence for restriction 
because a committed state will be turned into an uncommitted one if all its outgoing committed transitions are removed during the restriction process, which may violate the third condition of the timed step simulation. To address this, we provide the following theorem to guarantee that the timed step simulation preorder is a precongruence for restriction.
The corresponding proof is provided in 
Appendix~\ref{APP: proof of side condition}.

\begin{theorem} \label{Theorem: side condition}
    Let $\mathcal{T}_1$ and $\mathcal{T}_2$ be comparable TTSBs such that $\mathcal{T}_1 \preceq \mathcal{T}_2$. Let $C \subseteq \Delta \cup \mathcal{C}\cup E_1$. If for any committed state $r$ of $\mathcal{T}_1$, there exists $a\in Act - \{\delta?,c!,c?\mid \delta \in C \cap \Delta, c\in C\cap \mathcal{C}\}$ such that $r\xrightarrow{a,1}_1$, then $\mathcal{T}_1 \backslash C \preceq  \mathcal{T}_2 \backslash C$.
\end{theorem}

Intuitively, the side condition of Theorem~\ref{Theorem: side condition} ensures that a committed state in $\mathcal{T}_1$ is still committed in $\mathcal{T}_1 \backslash C$. This condition is not problematic in practice, as a well-defined timed system model should ensure that from any committed state, there is always an executable transition, which could be labeled with an input broadcast action $\delta!$
or internal action $\tau$, thereby satisfying the side condition.

\section{Compositional Verification}\label{sec: Compositional Verification}
This section shows how our theorems help reduce the state space in verifying timed systems with broadcast synchronization.  
We first present a verification framework for \emph{safety properties}, based on the theorems we develop. Since most timed automata model checkers, except U{\scriptsize PPAAL}, do not support non-blocking broadcast, existing benchmark suites are limited. To demonstrate the effectiveness of our framework, we apply it to two case studies: a producer-consumer system and the clock synchronization protocol from~\cite{new_case_try}.
All the experiments\footnote{All the U{\scriptsize PPAAL} models and raw experiment data for this paper are available at \url{https://github.com/zeno-98/CAV-2025-333}.} in this paper were conducted using the U{\scriptsize PPAAL} 5.0 tool on a 4.0 GHz AMD Ryzen 5 2600X processor with 32 GB of RAM, running 64-bit Windows 10, with a timeout of 3,600 seconds. 

\subsection{Verification Framework for Safety Properties}\label{Verification Framework for Safety Properties}

This paper focuses on verifying safety properties, a fundamental class of specifications in timed system verification. Intuitively, they assert that ``something bad never happens,'' capturing the absence of undesirable behaviors. We formally define the safety properties of NTA as follows.

\begin{definition}[Safety Properties of NTA]\label{Safety Properties of NTA}
    Let $\mathcal{N}=\langle\mathcal{A}_1,\dots,\mathcal{A}_n\rangle$ be an NTA, and $P$ be a property over a subset of $V=\bigcup_{i=1}^n (V_i\cup{\mathsf{loc}_i})$. We say that $P$ is a safety property of $\mathcal{N}$, notation $\mathcal{N}\models \forall \square P$, iff for all reachable states $s$ of $\mathsf{LTS}(\mathcal{N})$, $s\models P$.
\end{definition}

The following theorem states that, given a timed system and a property, we can replace some system components with their corresponding abstractions to obtain an abstract version of the original system. If the property is proven to be a safety property of the abstraction, it must also be a safety property of the original system. Naturally, the property should not depend on the internal variables or locations of the components to be abstracted, as they may be merged or even deleted during the abstraction process.

\begin{theorem}\label{theorem: verify safety properties}
    Let $\mathcal{N}=\langle\mathcal{A}_1,\dots,\mathcal{A}_i,\mathcal{A}_{i+1},\dots,\mathcal{A}_n\rangle$ and $\mathcal{N}'=\langle\mathcal{B}_1,\dots,\mathcal{B}_j,\mathcal{A}_{i+1},$ $\dots,\mathcal{A}_n\rangle$ be two NTAs and $P$ be a property over $\hat V=\bigcup_{k=i+1}^n (V_k\cup{\mathsf{loc}_k})$. 
    Let ${\hat E} = \bigcup_{k=1}^i E_k - \hat V$, $\mathcal{T}_a = (\mathsf{TTSB}(\mathcal{A}_1)\|\dots\|\mathsf{TTSB}(\mathcal{A}_i))\backslash({\Delta}\cup\mathcal{C}\cup{\hat E} -\Sigma(\mathcal{T}_c))$, $\mathcal{T}_b= (\mathsf{TTSB}(\mathcal{B}_1)\|\dots\|\mathsf{TTSB}(\mathcal{B}_j))\backslash({\Delta}\cup\mathcal{C}\cup{\hat E}-\Sigma(\mathcal{T}_c))$, and $\mathcal{T}_c = \mathsf{TTSB}(\mathcal{A}_{i+1})\|\dots$ $\|\mathsf{TTSB}(\mathcal{A}_n)$. If $P$ is a safety property of $\mathcal{N}'$, $\mathcal{T}_a$ and $\mathcal{T}_b$ are comparable with $\mathcal{T}_a\preceq \mathcal{T}_b$, and $\mathcal{T}_a\|\mathcal{T}_c$ satisfies the side condition of Theorem~\ref{Theorem: side condition} with $C={\Delta}\cup\mathcal{C}$, then $P$ is also a safety property of $\mathcal{N}_1$.
\end{theorem}

The proof of Theorem~\ref{theorem: verify safety properties} is in
Appendix~\ref{APP: proof of verify safety properties}.
By this theorem, we can check property $P$ is a safety property of NTA  $\mathcal{N}=\langle\mathcal{A}_1,\dots,\mathcal{A}_n\rangle$ in a compositional way using the following steps:
\begin{enumerate}

    \item 
    Partition $\mathcal{N}$ appropriately into two parts $\mathcal{A}_1,\dots,\mathcal{A}_i$ and $\mathcal{A}_{i+1},\dots,\mathcal{A}_n$, such that $P$ does not depend on the internal variables and locations of $\mathcal{A}_1,\dots,\mathcal{A}_i$.

    \item 
    Construct suitable TAs $\mathcal{B}_1,\dots,\mathcal{B}_j$, such that $(\mathsf{TTSB}(\mathcal{A}_1)\|\dots\|\mathsf{TTSB}(\mathcal{A}_i))\backslash C$ $\preceq (\mathsf{TTSB}(\mathcal{B}_1)\|\dots\|\mathsf{TTSB}(\mathcal{B}_j))\backslash C$, where $C$ is the set of broadcast channels, binary channels and external variables unused in $\mathcal{A}_{i+1},\dots,\mathcal{A}_n$.

    \item 
    Use model-checking tool U{\scriptsize PPAAL} to verify if $\mathcal{N'}\models \forall\square P$, where $\mathcal{N}'=\langle \mathcal{B}_1,\dots,\mathcal{B}_j,\mathcal{A}_{i+1},\dots,\mathcal{A}_n\rangle$. If it does, then by Theorem~\ref{theorem: verify safety properties}, $P$ is a safety property of $\mathcal{N}$. Otherwise, return to step 1 or 2 to try alternative partitioning methods or construct another suitable group of abstract automata. 

\end{enumerate}

\subsection{Case Study I: Producer-Consumer System}\label{Case Study I}
We first apply our framework to a producer-consumer system, which includes one producer, $\mathsf{N}$ consumers, and a coordinator. The producer generates data packets at fixed time intervals and stores the packets in a buffer for consumers to consume later. Each consumer obtains an exclusive right to consume a data packet. Once obtained, the consumer will either consume a packet and release the right or defer the consumption. The coordinator ensures that once a consumer obtains the right, it should consume a packet within a required interval. These components communicate via broadcast channels, binary channels, and a shared variable. In addition, the system model also has several committed locations. 
We manually constructed an abstraction that combines the coordinator and all consumers into a single timed automaton $\mathcal{A}$. This abstraction consists of five locations (including one committed location), five transitions, and one internal clock. It captures the behavior of consumers consuming data packets under the coordinator's control. In particular, it reflects all possible time intervals between two consecutive data consumption events. Notably, this abstraction is independent of the parameter $\mathsf{N}$, that is, $\mathcal{A}$ applies to any number of consumers and can simulate their external behavior together with that of the coordinator.
Due to the page limit, the NTA model, the abstracted model, and the corresponding proof are given in 
Appendix~\ref{APP: Detail of Case Study I}.


Because of the limited buffer size of the producer, given a certain parameter setting, we expect that there will be no overflow during system operation. We apply both the traditional monolithic verification ($\mathsf{MV}$), which directly 
checks the property of the original model and the compositional verification ($\mathsf{CV}$) described in Section~\ref{sec: Compositional Verification}.
Since U{\scriptsize PPAAL} is currently the only model checker supporting non-blocking broadcast synchronization, we use it exclusively to perform $\mathsf{MV}$.
The experimental results confirm that no overflow occurs, and Table~\ref{tab:example_PCS} presents the average verification time in seconds over five runs for different values of $\mathsf{N}$ with $\mathsf{MV}$ and $\mathsf{CV}$. 

\begin{table}[h]
    \vspace{-6mm}
    \caption{Verification time of no overflow}
    \vspace{2mm}
    \centering
    
    \begin{tabular}{P{3pt} c P{5pt} c P{5pt} c P{5pt} c P{5pt} c P{5pt} c P{5pt} c P{5pt} c P{3pt}} 
        \toprule
        &$\mathsf{N}$ && \textbf{9} && \textbf{10} && \textbf{11} && \textbf{12} && \textbf{13} && \textbf{14} && \textbf{15} &  \\
        \midrule
        &$\mathsf{MV}$ && 3.769 && 11.778 && 36.767 && 109.843 && 412.710 && 1656.301 && timeout & \\
        &$\mathsf{CV}$ && 0.005 && 0.005 && 0.005 && 0.005 && 0.005 && 0.005 && 0.005 &\\
        \bottomrule
    \end{tabular}

    \label{tab:example_PCS}
    \vspace{-3mm}
\end{table}

The $\mathsf{MV}$ row shows that the verification time required by the traditional monolithic method grows exponentially as $n$ increases and exceeds $3,600$ seconds when $\mathsf{N}=15$. The $\mathsf{CV}$ row presents the verification time by our compositional verification method, which implies that $\mathsf{CV}$ outperforms $\mathsf{MV}$ significantly.
Since the abstraction we built simulates the compositional behaviors of the coordinator and all the consumers for any $\mathsf{N} \ge 1$, we obtain the verification results for the system with an arbitrary number of consumers in 5 milliseconds.


\subsection{Case Study II: Clock Synchronization Protocol}\label{Case Study II}
Secondly, we turn to the clock synchronization protocol presented in~\cite{new_case_try} as a case study. The Dutch company Chess develops this protocol to address a critical challenge in designing wireless sensor networks (WSNs): the hardware clocks of sensors in the network may drift. So, ensuring clock synchronization of the protocol is vital to guarantee communication reliability in the networks.

The NTA model of the protocol consists of $\mathsf{N}$ nodes, named $0,\dots,\mathsf{N}-1$. 
These nodes take turns broadcasting messages to the others in a fixed order to perform clock synchronization. After completing one round, they wait for a specific period and start the next cycle. Each node internally contains three TAs: \textbf{Clock}, \textbf{WSN}, and \textbf{Synchronizer}, which communicate with each other through broadcast channels and shared variables.
Automaton \textbf{Clock} models the node's hardware clock, which may drift, 
automaton \textbf{WSN} takes care of broadcasting messages, 
and automaton \textbf{Synchronizer} resynchronizes the hardware clock upon receipt of a message. 
As can be seen, in the designed model, two types of broadcast synchronization should exist: the internal type in each node and the external type among the nodes.  The NTA model provided in~\cite{new_case_try} and the corresponding abstractions we build are in 
Appendix~\ref{APP: Detail of Case Study II}. 

Given a certain parameter setting, we apply both $\mathsf{MV}$ and $\mathsf{CV}$ methods to check whether the NTA model satisfies the property: the hardware clocks of all nodes remain synchronized during network operation. The experimental results show that this property is satisfied, and  Table~\ref{tab:example_WSN} presents the average verification time in seconds over five runs for different values of $\mathsf{N}$.


\begin{table}[h]
    \vspace{-6mm}
    \caption{Verification time of hardware clocks synchronization}
    \vspace{2mm}
    \centering
    \begin{tabular}{P{3pt} c P{10pt} c P{10pt} c P{10pt} c P{10pt} c P{10pt} c P{3pt}} 
        \toprule
        &\textbf{$\mathsf{N}$} && \textbf{3} && \textbf{4} && \textbf{5} && \textbf{6} && \textbf{7} & \\
        \midrule
        &$\mathsf{MV}$ && 0.070 && 2.256 && 185.641 && timeout && timeout &\\
        
        &$\mathsf{CV}$ && 0.303 && 0.936 && 2.230 && 4.309 && 7.629 & \\
        \bottomrule
    \end{tabular}
    
    \label{tab:example_WSN}
    \vspace{-4mm}
\end{table}



The $\mathsf{MV}$ row shows that the verification time required by the traditional monolithic method grows significantly as $\mathsf{N}$ increases. To apply $\mathsf{CV}$, we first select two nodes, $a$ and $b$, from the $\mathsf{N}$ nodes with $0\le a<b\le \mathsf{N}-1$, and abstract the remaining $\mathsf{N}-2$ nodes into a single TA $\mathcal{A}$. 
This abstraction is also constructed manually. It contains three locations, $4\mathsf{N}$ transitions, and an internal clock. It abstracts away the details of how the $\mathsf{N}-2$ intermediate nodes handle received synchronization messages, as well as the specific order in which they broadcast messages during a round. Instead, it ensures that exactly $\mathsf{N}$ clock synchronization events occur in each round and captures the possible time intervals between two consecutive clock synchronization events.
Based on this, we check whether or not the abstracted NTA composed of the models of nodes $a$, $b$ and the abstraction $\mathcal{A}$, satisfies the property that the hardware clocks of the two selected nodes always remain synchronized. Clearly, if the property is satisfied for all the choices of $a$ and $b$, we conclude that the original system satisfies the target property. 
Note that the system is not strictly symmetric, as nodes broadcast messages periodically in a fixed order and different choices of $a$ and $b$ result in different time intervals between their broadcast actions. Therefore, we must enumerate all possible pairs of $a$ and $b$, and the total verification time of our compositional method is the sum of the checking times for all these cases.
For instance, when $\mathsf{N}=6$, we need to verify $C_6^2=\frac{(6\times5)}{2} = 15$ different cases. The $\mathsf{CV}$ row demonstrates the total verification time required by our compositional verification method for each $\mathsf{N}$. Although in the case of $\mathsf{N}=3$, our method takes a slightly longer time since $\mathcal{A}$ has more behaviors than a single node, it demonstrates significant efficiency advantages when $\mathsf{N}\ge5$.

\section{Conclusion} \label{sec: conclusion}

This paper proposes the first compositional abstraction framework for timed systems with broadcast synchronization, providing a method to reduce state space in model checking. Specifically, this framework focuses on timed systems modeled as NTAs in U{\scriptsize PPAAL}, and also supports binary synchronization, shared variables, and committed locations. For this purpose, we first define TTSB, which extends LTSs with state variables, transition commitments, and time-related behaviors, along with corresponding parallel composition and restriction operations. We prove that the parallel composition operator is both commutative and associative. Secondly, we provide compositional and non-compositional semantics for NTAs with broadcast synchronization in U{\scriptsize PPAAL} and prove their equivalence. Thirdly, we define the timed step simulation relation for TTSBs and prove it is a precongruence for parallel composition. Finally, we demonstrate that safety properties verified on abstractions are preserved in the original models and validate the efficiency of our framework through two case studies.
Future work includes extending the framework to support other U{\scriptsize PPAAL} features, such as urgent channels and priorities. We also plan to integrate abstraction refinement methods~\cite{Abstraciton_Refinement} to develop an automated compositional verification workflow similar to those in~\cite{Chy_automatic,Kim_automatic}.
This would enable the application of our compositional abstraction framework to a broader range of real-world cases, such as the timing-based broadcast algorithms discussed in~\cite{Nancy_Book}.


\newpage
\bibliographystyle{splncs04}
\bibliography{reference}

\appendix
\newpage
\section{Proof of Lemma~\ref{lemma: Composition well-defined}} \label{APP: proof of Composition well-defined}

For proving Lemma~\ref{lemma: Composition well-defined}, that is, the composition of two TTSBs is still a TTSB, the following lemma is necessary: a state of the composition is committed iff one of its component states is committed. 

\begin{lemma} \label{lemma: committed states}
    Let $\mathcal{T}_1$ and $\mathcal{T}_2$ be compatible TTSBs. Let $r\in S_1$ and $s\in S_2$ such that $r\heartsuit s$, then $Comm(r\|s)\Leftrightarrow Comm(r)\vee Comm(s)$.
\end{lemma}

\begin{proof}
    \begin{enumerate}
        \item[$\Rightarrow$]
        If $Comm(r \| s)$ holds, there must be a committed transition of the form $r \| s \xrightarrow{a, 1} r' \| s'$. As this transition is committed, it can be generated by rules other than \textbf{TIME}. For the committed transition generated by  \textbf{EXT}, or \textbf{TAU} or \textbf{SYNC}, we refer to~\cite{Compositional_Abstraction} for the detailed proof that  $Comm(r) \vee Comm(s)$ holds. 
        For the transition generated by \textbf{SND}, we assume w.l.o.g. $i = 1$, there exist transitions $r \xrightarrow{\delta!, b}_1 r'$ and $s[r'] \xrightarrow{\delta?, b'}_2 s'$ with $b \vee b'$, and by Axiom~III, $s[r'] \xrightarrow{\delta?, b'}_2$ implies $s \xrightarrow{\delta?, b'}_2$. Similarly, for the transition generated by \textbf{RCV}, there exist  $r \xrightarrow{\delta?, b}_1 r'$ and $s \xrightarrow{\delta?, b'}_2 s'$ with $b \vee b'$. Hence, we have $Comm(r)\vee Comm(s)$.
        
        \item[$\Leftarrow$]
        If $Comm(r)\vee Comm(s)$ holds, we assume w.l.o.g. $Comm(r)=1$. By Axiom~I, there must exist a committed transition $r\xrightarrow{a,1}_1 r'$, and $a \in \mathcal{E}_{\Delta} \cup \mathcal{E}_{\mathcal{C}} \cup\{\tau\}$. If $a \in \mathcal{E}_{\Delta}$, that is, $a = \delta!$ or $a = \delta?$, according to Axiom~V, there exists a transition $s[r'] \xrightarrow{\delta?, b'}_2 s'$ or $s \xrightarrow{\delta?, b'}_2 s'$ in $\mathcal{T}_2$. Further by rule \textbf{SND} or \textbf{RCV}, we can establish $r\|s\xrightarrow{a,1}r'\|s'$ and $Comm(r\|s)$. If $a \in \mathcal{E}_{\mathcal{C}}$ or $a = \tau$, we can respectively use rule \textbf{EXT} or rule \textbf{TAU} to establish transition $r\|s \xrightarrow{a,1}$ to satisfy $Comm(r\|s)$.     
        $\hfill\square$
    \end{enumerate} 
\end{proof}

Now, we can prove Lemma~\ref{lemma: Composition well-defined}.
Let $E = E_1\cup E_2$, $H = H_1\cup H_2$. Since $E_1\cap H_1 = E_2\cap H_2 =\emptyset$ ($\mathcal{T}_1$ and $\mathcal{T}_2$ are TTSBs), and $E_1\cap H_2 = E_2\cap H_1 = \emptyset$ ($\mathcal{T}_1$ and $\mathcal{T}_2$ are compatible), we have $E \cap H=\emptyset$. Let $S$ be the composed state set and $s_0$ the initial state.   By definition~\ref{def: Parallel composition}, $S\subseteq {V\!al}(V)$ and $s^0\in S$. We must prove that $\mathcal{T}_1\|\mathcal{T}_2$ still satisfies the six axioms for a TTSB.
    Suppose $r,r'\in S_1$, $s,s' \in S_2$ and $r\heartsuit s$. 
    \begin{enumerate}
        \item 
        Assume that $r\|s \xrightarrow{a,1}\wedge\: r\|s \xrightarrow{a',b}$. To prove the satisfaction of Axiom~I, we must establish $a'\in \mathcal{E}_{\Delta}\cup \mathcal{E}_{\mathcal{C}} \vee (a'=\tau \wedge b)$. Since the considered action set is $Act \triangleq \mathcal{E}_{\Delta}\cup \mathcal{E}_{\mathcal{C}}\cup\{\tau\}\cup\mathbb{R}_{\geq 0}$, we can establish this by demonstrating the following two parts. 
        \begin{itemize}
            \item $a'\not\in \mathbb{R}_{\geq 0}$, i.e. $a'$ is not a time-passage action. Since $r\|s\xrightarrow{a,1}$, we have $Comm(r\|s)$, which implies $Comm(r)\vee Comm(s)$ by Lemma~\ref{lemma: committed states}.
            Then by Axiom~I for $\mathcal{T}_1$ and $\mathcal{T}_2$, we have either $r$ or $s$ does not have outgoing time-passage transitions. Finally by rule \textbf{TIME}, $r\|s$ has no outgoing time-passage transitions, that is, $a'\not\in \mathbb{R}_{\geq 0}$.
            
            \item $a'=\tau \Rightarrow b$\footnote{It equals to proving that if $a'=\tau$, $b$ will not be $0$.}. If $a'=\tau$, then either rule \textbf{TAU} or rule \textbf{SYNC} is used to prove $r\|s\xrightarrow{a',b}$. As mentioned above, $Comm(s)\vee Comm(r)$ follows from $r\|s\xrightarrow{a,1}$. Assume w.l.o.g. that $i = 1$. If rule \textbf{TAU} is used and $Comm(r)$, then $b=1$ by Axiom~I for $\mathcal{T}_1$. If rule \textbf{TAU} is used and $Comm(s)$, then $b=1$ since rule \textbf{TAU} has the condition $Comm(s)\Rightarrow b$. If rule \textbf{SYNC} is used, then $Comm(r)\vee Comm(s)$ implies $b=1$. Hence, we can conclude that $a'=\tau \Rightarrow b$ holds.
            
        \end{itemize}
        
        \item 
        Let $r\|s\in S$ and $u \in V\!al(E)$. Similar to the proof of ~\cite{Compositional_Abstraction}, we have $(r\|s)[u]\in S$, which implies that $\mathcal{T}_1\|\mathcal{T}_2$ satisfies Aixom~II.
    
        \item 
        To prove the satisfaction of Axiom III, suppose $r\|s\xrightarrow{\sigma?,b}$ and $u\in V\!al(E)$. We must establish that $(r\|s)[u]\xrightarrow{\sigma?,b}$. Note here, the channel $\sigma$ can be binary or broadcast. If $\sigma \in \mathcal{C}$, we can establish this by referring to the proof in~\cite{Compositional_Abstraction}. If $\sigma \in \Delta$, $r\|s\xrightarrow{\sigma?,b}$ must be proved by rule \textbf{RCV}. This implies that $r\xrightarrow{\sigma?,b_1}$ and $s\xrightarrow{\sigma?,b_2}$, where $b=b_1\vee b_2$. Then by Axiom~III, we have $r[u]\xrightarrow{\sigma?,b_1}$ and $s[u]\xrightarrow{\sigma?,b_2}$. Finally, by rule \textbf{RCV}, we obtain $(r\|s)[u]\xrightarrow{\sigma?,b}$. Hence, $\mathcal{T}_1 \| \mathcal{T}_2$ satisfies Axiom~III.
        
        \item
        Axiom IV for $\mathcal{T}_1\|\mathcal{T}_2$ follows trivially from Axiom~IV for $\mathcal{T}_1$ and $\mathcal{T}_2$ and rule \textbf{TIME} of the parallel composition.
    
        \item 
        To prove the satisfaction of Axiom~V,  for any broadcast channel $\delta \in \Delta$, we must establish that $r\|s\xrightarrow{\delta?, b}$ holds. 
        By Axiom~V, $\mathcal{T}_1$ and $\mathcal{T}_2$ must respectively have transitions in the form of $r\xrightarrow{\delta?,b}r'$ and $s\xrightarrow{\delta?,b'}s'$. Then by rule \textbf{RCV}, we can obtain that  $\mathcal{T}_1\|\mathcal{T}_2$ has a transition $r\|s\xrightarrow{\delta?,b\vee b'}r'\|s'$. Hence, $\mathcal{T}_1 \| \mathcal{T}_2$ satisfies Axiom~V.
        
        \item 
        To prove the satisfaction of Axiom VI, suppose $\mathcal{T}_1\|\mathcal{T}_2$ has a transition $r\|s\xrightarrow{\delta?,b}r'\|s'$. We must establish that $r\|s\lceil E = r'\|s'\lceil E$. The transition $r\|s\xrightarrow{\delta?,b}r'\|s'$ is generated by the parallel composition's rule \textbf{RCV}. This implies that $\mathcal{T}_1$ and $\mathcal{T}_2$ respectively have transitions $r\xrightarrow{\delta?,b_1}r'$ and $s\xrightarrow{\delta?,b_2}s'$, where $b=b_1\vee b_2$. Then by Axiom~VI for $\mathcal{T}_1$ and $\mathcal{T}_2$, we have $r\lceil E_1 = r' \lceil E_1$ and $s\lceil E_2 = s' \lceil E_2$. Further by $E_2\cap H_1 = \emptyset$ and $E=E_1 \cup E_2$, we have $r\lceil E = r' \lceil E$. Similarly, we have $s\lceil E = s' \lceil E$. Clearly, we have  $(r\|s)\lceil E=r\lceil E \| s\lceil E = r'\lceil E \| s'\lceil E = (r'\|s')\lceil E$. Hence, $\mathcal{T}_1 \| \mathcal{T}_2$ satisfies Axiom~VI. 
    \end{enumerate}

\section{Proof of Theorem~\ref{theroem: Commutativity and Associativity}} \label{APP: proof of Commutativity and Associativity}
If $\mathcal{T}_1$ and $\mathcal{T}_2$ are compatible TTSBs, then the commutativity, i.e., $\mathcal{T}_1\|\mathcal{T}_2$ = $\mathcal{T}_2\|\mathcal{T}_1$, can be directly obtained from the symmetry of the parallel composition rules. Below, we present the proof of associativity.

If $\mathcal{T}_1$, $\mathcal{T}_2$ and $\mathcal{T}_3$ are pairwise compatible TTSBs, then by Lemma~\ref{lemma: expressions}(2), we have $\mathcal{T}_1\|\mathcal{T}_2$ is compatible with $\mathcal{T}_3$, and $\mathcal{T}_1$ is compatible with $\mathcal{T}_2\|\mathcal{T}_3$. Let $\mathcal{T}_L=(\mathcal{T}_1\|\mathcal{T}_2)\|\mathcal{T}_3$ and $\mathcal{T}_R = \mathcal{T}_1\|(\mathcal{T}_2\|\mathcal{T}_3)$. It is easy to find that $\mathcal{T}_L$ and $\mathcal{T}_R$ agree on 5 components ($E,H,S,s^0,Act$), except for the transition sets ($\rightarrow ^1,\rightarrow ^0$). Since the action set is $Act \triangleq \mathcal{E}_{\Delta}\cup\mathcal{E}_{\mathcal{C}}\cup\{\tau\}\cup\mathbb{R}_{\geq 0}$, we can prove $\mathcal{T}_L$ and $\mathcal{T}_R$ share the same transition sets, that is, the sets of time-passage transition, $\tau$-transition, binary transition and broadcast transition. For the former 3 transition sets, the proof process is similar to that in~\cite{Compositional_Abstraction}, except that the newly defined Lemma~\ref{lemma: committed states} is used.
For the $4^{th}$ transition set, we will distinguish 4 cases to prove that the broadcast transition set of $\mathcal{T}_R$ contains that of $\mathcal{T}_L$. The four cases correspond to different broadcast scenarios, i.e., each component has $\delta?$ (1 case) and one component has $\delta!$ (3 cases).
Then, by a symmetric argument, the reverse inclusion can also be obtained. Still, let  $r,r'\in S_1, s,s'\in S_2, t,t'\in S_3$ and  $r\heartsuit s\heartsuit t$. 

\begin{itemize}
    \item Case $(\delta?\ \delta?\ \delta?)$. 
    In this case, $r\xrightarrow{\delta?,b}_1 r'$, $s\xrightarrow{\delta?,b'}_2 s'$ and $t\xrightarrow{\delta?,b''}_3 t'$. According to the definition of rule \textbf{RCV} and the associativity of $\|$ and $\vee$, the transition generated in right-associative order is the same as the one generated in left-associative order, i.e., they are both $r\|s\|t\xrightarrow{\delta?,b\vee b' \vee b''}r'\|s'\|t'$.

    \item Case $(\delta!\ \delta?\ \delta?)$.
    In this case, $r\xrightarrow{\delta!,b}_1 r'$, $s[r']\xrightarrow{\delta?,b'}_2 s'$ and  $t[r'\|s']\xrightarrow{\delta?,b''}_3 t'$. In the left-associative order, we first get transition $r\|s\xrightarrow{\delta!,b\vee b'} r'\|s'$ by rule \textbf{SND}. Again by this rule, we get $r\|s\|t\xrightarrow{\delta!,b\vee b'\vee b''} r'\|s'\|t'$. In the right-associative order, we have:
        \begin{align*}
            t[r'\|s']
            =\:&
            t[(r'\|s')\lceil E_3] 
            &\text{Pairwise Compatible}
            \\
            =\:&
            t[(r'\lceil E_3)\|(s'\lceil E_3)] 
            &\text{Obviously}
            \\
            =\:&
            t[(r'\lceil E_3)\|(s[r']\lceil E_3)] 
            &\text{Axiom VI}
            \\
            =\:&
            t[r'\|s[r']] 
            &\text{Arrangement}
            \\
            =\:&
            t[r'\rhd s] 
            &\text{Lemma~\ref{lemma: expressions} (3)}
            \\
            =\:&
            t[s][r'] 
            &\text{Lemma~\ref{lemma: expressions} (4)}
            \\
            =\:&
            t[r'] 
            &\text{$s$ and $t$ are compatible}
        \end{align*}
    This means that we can merge $s[r']\xrightarrow{\delta?,b'}_2 s'$ and $t[r'\|s']\xrightarrow{\delta?,b''}_3 t'$ into $(s\|t)[r']$ $\xrightarrow{\delta?, b'\vee b''} s'\|t'$ by rule \textbf{RCV}. Further by rule \textbf{SND}, we get the same transition $r\|s\|t\xrightarrow{\delta!,b\vee b' \vee b''} r'\|s'\|t'$.

    \item Case $(\delta?\ \delta!\ \delta?)$. 
    In this case, $r[s']\xrightarrow{\delta?,b}_1 r'$, $s\xrightarrow{\delta!,b'}_2 s'$ and $t[r'\|s']\xrightarrow{\delta?,b''}_3 t'$. In the left-associative order, by using rule \textbf{SND} twice, we can first establish transition $r\|s\xrightarrow{\delta!,b\vee b'} r'\|s'$, and then $r\|s\|t\xrightarrow{\delta!,b\vee b'\vee b''} r'\|s'\|t'$. In the right-associative order, similar to case $(\delta!\ \delta?\ \delta?)$, we can prove $r[s'\|t']=r[s']$ and $t[s']=t[r'\|s']$. Then by rule \textbf{SND}, we get transition
    $s\|t\xrightarrow{\delta!,b'\vee b''} s'\|t'$ from transitions $s\xrightarrow{\delta!,b'}_2 s'$ and $t[r'\|s']\xrightarrow{\delta?,b''}_3 t'$. Finally, again by rule \textbf{SND}, we 
    get the same transition $r\|s\|t\xrightarrow{\delta!,b\vee b' \vee b''} r'\|s'\|t'$.

    \item Case $(\delta?\ \delta?\ \delta!)$. The proof is similar to that of case $(\delta!\ \delta?\ \delta?)$.
\end{itemize}

\section{Proof of Lemma~\ref{lemma: TTSB(A) is a TTSB}} \label{APP: proof of TTSB(A) is a TTSB}

    Since $\mathcal{A}$ is a TA, $E$ and $H$ are disjoint. Additionally, since the variable $\mathsf {loc}$ is fresh, $E$ and $H \cup \{\mathsf{loc}\}$ are also disjoint. By the definition of TA, we have $v^0 \models I(l^0)$, which implies that $s^0 \in S$, as required. Now, we check that $\mathsf{TTSB}(\mathcal{A})$ satisfies all six axioms for a TTSB:
    
    \begin{itemize}
    \item 
    Suppose there is a state $s \in S$ with $s \xrightarrow{a, 1} s'$ and $s \xrightarrow{a', b} s''$. Since the committed transition $s \xrightarrow{a, 1} s'$ can only be generated by rule \textbf{ACT}, it follows that $s(\mathsf{loc}) \in K$. By Definition~\ref{def: Timed step simulation for TTSBs}, this implies that $s \xrightarrow{a', b} s''$ can only be generated by rules \textbf{ACT} or \textbf{VIRT}. If it is generated by rule \textbf{ACT}, then $b = 1$; if by rule \textbf{VIRT}, then $a' \in \mathcal{E}_{\Delta}$. So Axiom~I is satisfied.
    
    \item 
    Axiom~II is satisfied because, according to Axiom~VII, all location invariants do not depend on external variables.
    
    \item 
    Axiom~III is satisfied because, according to Axiom~VIII, all input guards do not depend on external variables.
    
    \item Axiom~IV is obtained immediately from rule \textbf{TIME}.
    
    \item 
    For a broadcast action $\delta?$ and state $s \in S$ with $s(\mathsf{loc})=l$, if $\mathcal{A}$ has the transition $l\xrightarrow{g,\delta?,\rho}l'$ with $s \models g$, the transition $s\xrightarrow{\delta?,b}s'$ will be generated by rule \textbf{ACT}. Otherwise, the transition $s\xrightarrow{\delta?,0}s$ is be generated by rule \textbf{VIRT}. This ensures that Axiom V is satisfied.
    
    \item 
    Axiom~VI is clearly satisfied because, according to Axiom~XI, external variables will not be updated in $\delta?$-transitions.
    \end{itemize}

\section{Proof of Theorem~\ref{Theorem: semantics equal}} \label{APP: proof of semantics equal}
For proving Theorem~\ref{Theorem: semantics equal}, the following lemma is necessary: in the TTSB semantics of a TA, a state is committed iff the corresponding location is committed.

\begin{lemma} \label{lemma: K implies Comm(s)}
    Let $\mathcal{A}$ be a TA and let $s$ be a state of $\mathsf{TTSB}(\mathcal{A})$. Then $s(\mathsf{loc})\in K \Leftrightarrow Comm(s)$.
\end{lemma}

\begin{proof}
    This can be directly obtained through Axiom~IX and rule \textbf{ACT} in Definition~\ref{def: TTSB semantics of TA}, and the details are not elaborated here.
    $\hfill\square$
\end{proof}

Now, we demonstrate that the compositional and non-compositional semantics of NTA with broadcast channels are equivalent. It follows directly from definitions (\ref{def: TTSB}, \ref{def: Parallel composition}, \ref{def: Restriction for TTSB}, \ref{def: LTS semantics of NTA} and \ref{def: TTSB semantics of TA}) that both sides of the equation have the same set of states and the same initial state. The transition set of $\mathsf{LTS}(\mathcal{N})$ can be divided into the following three subsets.
    \begin{samepage}
    \begin{enumerate}
        \item 
        The set of $\tau$-transitions generated by rules \textbf{BCST}.
        \item 
        The set of $\tau$-transitions generated by rules \textbf{TAU} and \textbf{SYNC}.
        \item 
        The set of time-passage transitions generated by rule \textbf{TIME}.
    \end{enumerate}
    \end{samepage}
    Similarly, the transition set of $\mathsf{LTS}((\mathsf{TTSB}(\mathcal{A}_1)\|\cdots\|\mathsf{TTSB}(\mathcal{A}_n))\backslash (\Delta \cup \mathcal{C}))$ can also be divided into three subsets.
    \begin{enumerate}
        \item
        The set of $\tau$-transitions introduced to replace $\delta!$-transitions during the restriction process.
        \item 
        The set of other $\tau$-transitions.
        \item 
        The set of time-passage transitions.
    \end{enumerate}
    We can follow the proof in \cite{Compositional_Abstraction} to show that the transition sets $H_2$ and $H_3$ correspondingly equals $C_2$ and $C_3$, except that the newly proved Lemma~\ref{lemma: committed states} is used. Now we prove that transition set $H_1$ is equivalent to $C_1$. Let $\overline{RS}(\delta,i,s) = \{1\leq k \leq n,k \not = i, k\not\in RS(\delta,i,s)\}$, and $k_1,k_2,\dots,k_{\bar{m}}$ enumerate its elements. Obviously, $m+\bar{m}+1=n$.
    In the rest of this paper, if there is no ambiguity, we abbreviate $RS(\delta,i,s)$ and $\overline{RS}(\delta,i,s)$ as $RS$ and $\overline{RS}$ respectively, and $s\lceil W_q$ as $s_q$, for $1\leq q \leq n$.
    
    \begin{description}
        \item $\subseteq$ 
        Assume $\mathsf{LTS}(\mathcal{N})$ has a transition $s\xrightarrow{\tau}s'$ generated by rule \textbf{BCST} in Fig.\ref{fig:NTA_UPPAAL_Semantics}. Clearly, $s_1,s_2,\dots, s_n$ are pairwise compatible, and all the following hold.
        \begin{equation}
            \begin{array}{c}
                l_i\xrightarrow{g_i,\delta!,\rho_i}l'_i \myHspace\myHspace\myHspace\myHspace
                s(\mathsf{loc}_i)=l_i \myHspace\myHspace\myHspace\myHspace
                s\models g_i 
                \\
                \forall j\in RS:l_j\xrightarrow{g_j,\delta?,\rho_j}l_j',s(\mathsf{loc}_j)=l_j, s\models g_j 
                \\
                (\forall q:s(\mathsf{loc}_q)\not \in K_q)\vee l_i\in K_i\vee(\exists j\in RS: l_j\in K_j)
                \\
                s'= \rho_{j_m}(\cdots\rho_{j_1}(\rho_i(s)))[\{\mathsf{loc}_i\mapsto l_i', \mathsf{loc}_{j_1}\mapsto l_{j_1}',\dots,\mathsf{loc}_{j_m}\mapsto l_{j_m}'\}]
            \end{array}
            \nonumber
        \end{equation}
        We first prove that $\mathsf{TTSB}(\mathcal{A}_1),\cdots,\mathsf{TTSB}(\mathcal{A}_n)$ respectively contain the following transitions.
        \begin{itemize}
            \item 
            For $i$, let $s_i'=\rho(s_i)[\{\mathsf{loc}_i\mapsto l_i'\}]$ and $b_i\Leftrightarrow (l_i\in K_i)$. By rule \textbf{ACT}, $\mathsf{TTSB}(\mathcal{A}_i)$ contains transition $s_i\xrightarrow{\delta!,b_i}s_i'$.
            
            \item 
            For any $j\in RS$, by Axiom~X, $g_j$ does not depend on $E_j$, therefore $s_j[s_i']\models g_j$. Furthermore, $l_j\xrightarrow{g_j,\delta?,\rho_j}l_j'$ and clearly $s_j[s_i'](\mathsf{loc}_j)=l_j$. Let $s_j' =  \rho_j(s_j[s_i'])[\{\mathsf{loc}_j\mapsto l_j'\}]$ and $b_j\Leftrightarrow(l_j\in K_j)$. Then by rule \textbf{ACT}, $\mathsf{TTSB}(\mathcal{A}_j)$ contains transition $s_j[s_i']\xrightarrow{\delta?,b_j}s_j'$.
    
            \item 
            For any $k\in \overline{RS}$, by Axiom~IX, $s_k[s_i]$ does not satisfy the guard of any $\delta?$-transition starting from $s(\mathsf{loc}_k)$. Hence, by rule \textbf{VIRT}, $\mathsf{TTSB}(\mathcal{A}_k)$ contains transition $s_k[s_i]\xrightarrow{\delta?,0}s_k[s_i]$.
        \end{itemize}
        Given that the parallel composition operation is associative, we first repeatedly apply rule \textbf{RCV} to merge all the $\delta?$-transitions mentioned above. The result is the following transition.
        \begin{align*}
            s_{j_1}[s_i']&\|\cdots\|s_{j_m}[s_i']\|s_{k_1}[s_i']\|\cdots\|s_{k_{\bar{m}}}[s_i']
            \\
            &\xrightarrow{\delta?, b_{j_1}\vee\cdots\vee b_{j_m}} s_{j_1}'\|\cdots\|s_{j_m}'\|s_{k_1}[s_i']\|\cdots\|s_{k_{\bar{m}}}[s_i']
        \end{align*}
        Then, by rule \textbf{SND}, the following transition is generated.
        \begin{equation*}
            s
            \xrightarrow{\delta!,b_i\vee b_{j_1} \vee\cdots\vee b_{j_m}}
            s_i\|s_{j_1}'\|\cdots\|s_{j_m}'\|s_{k_1}[s_i']\|\cdots\|s_{k_{\bar{m}}}[s_i']
        \end{equation*}
       
        Then we prove that $s'=s_i\|s_{j_1}'\|\cdots\|s_{j_m}'\|s_{k_1}[s_i']\|\cdots\|s_{k_{\bar{m}}}[s_i']$, here we employ mathematical induction.
        \begin{itemize}
            \item Case $\lvert RS\rvert = 0$. In this case, $l_i\xrightarrow{g_i,\delta!,\rho_i}l_i'$ does not synchronize with any other transition, we have
            \begin{align}
                s'
                =\:&\rho_i(s)[\{\mathsf{loc}_i\mapsto l_i'\}] &\text{Rule \textbf{SND} of Definition~\ref{def: LTS semantics of NTA}}
                \nonumber
                \\
                =\:&s[\rho_i(s\lceil V_i)][\{\mathsf{loc}_i\mapsto l_i'\}] & \text{Definition of $\rho_i$}
                \nonumber
                \\
                =\:&s[\rho_i(s\lceil V_i) \lhd \{\mathsf{loc}_i\mapsto l_i'\}] & \text{Lemma~\ref{lemma: expressions}(4)}
                \nonumber
                \\
                =\:&s[\rho_i(s\lceil W_i)[\{\mathsf{loc}_i\mapsto l_i'\}]] &\text{Definitions of $\rho_i$ and $s_i$}
                \nonumber
                \\
                =\:&s[\rho_i(s_i)[\{\mathsf{loc}_i\mapsto l_i'\}]] &\text{Definition of $s_i$}
                \nonumber
                \\
                =\:& s[s_i'] & \text{Definition of $s_i'$}
                \nonumber
                \\
                =\:& s_i'\|s_{k_1}[s_i']\|\cdots\|s_{k_{\bar{m}}}[s_i'] &\text{Expansion of $s$}
                \nonumber
            \end{align}
           
            \item Case $\lvert RS\rvert = 1$. In this case, $l_i\xrightarrow{g_i,\delta!,\rho_i}l_i'$ only synchronizes with one transition, w.l.o.g, it is $l_j\xrightarrow{g_j,\delta?,\rho_j}l_j'$. Similar to the proof of rule \textbf{SYNC} in \cite{Compositional_Abstraction}, we can conclude $s'=s[s_i'\lhd s_j']$ and continue on this basis:
            \begin{align}
               s'
               =\:& s[s_i'\lhd s_j']
               &\text{From \cite{Compositional_Abstraction}}
               \nonumber
               \\
               =\:& s[s_i'\|s_j'] & \text{Axiom~VI}
               \nonumber
               \\
               =\:& s_i'\|s_j'\|s_{k_1}[s_i'\|s_j']\|\cdots\|s_{k_{\bar{m}}}[s_i'\|s_j'] & \text{Expansion of $s$}
               \nonumber
               \\
               =\:& s_i'\|s_j'\|s_{k_1}[s_i']\|\cdots\|s_{k_{\bar{m}}}[s_i']  & \text{Axiom~VI}
               \nonumber
            \end{align}
    
            \item Case $\lvert RS\rvert=q$. Assuming $s'=s_i'\|\cdots\|s_{j_q}'\|s_{k_1}[s_i']\|\cdots$ $\|s_{k_{\bar{m}}}[s_i']$ holds.
           
            \item Case $\lvert RS\rvert=q+1$. In this case:
            \begin{align*}
                s'
                =\:& \rho_{j_{q+1}}(\rho_{j_q}(\cdots\rho_{j_1}(\rho_i(s))))
                \\
                &[\{\mathsf{loc}_i\mapsto l_i', \mathsf{loc}_{j_1}\mapsto l_{j_1}',\dots,\mathsf{loc}_{j_q}\mapsto         l_{j_q}',\mathsf{loc}_{j_{q+1}}\mapsto l_{j_{q+1}}'\}] 
                \\
                &\text{Rule \textbf{SYNC} of Definition~\ref{def: LTS semantics of NTA}}
                \\
                =\:& \rho_{j_{q+1}}(\rho_{j_q}(\cdots\rho_{j_1}(\rho_i(s)))[\{\mathsf{loc}_i\mapsto l_i', \mathsf{loc}_{j_1}\mapsto     l_{j_1}',\dots,\mathsf{loc}_{j_q}\mapsto l_{j_q}'\}])
                \\
                &[\{\mathsf{loc}_{j_{q+1}}\mapsto l_{j_{q+1}}'\}] 
                \\
                &\text{Disjoint domains and reordering}
                \\
                =\:& \rho_{j_{q+1}}(s_i'\|s_{j_1}'\|\cdots\|s_{j_q}'\|s_{j_{q+1}}'[s_i']\|s_{k_1}[s_i']\|\cdots\|s_{k_{\bar{m}}}[s_i'])[\{\mathsf{loc}_{j_{q+1}}\mapsto l_{j_{q+1}}'\}]
                \\
                &\text{The assumption}
                \\ 
                =\:& s_i'\|s_{j_1}'\|\cdots\|s_{j_q}'\|s_{k_1}[s_i']\|\cdots\|s_{k_{\bar{m}}}[s_i']\|\rho_{j_{q+1}}(s_{j_{q+1}}'[s_i'])[\{\mathsf{loc}_{j_{q+1}}\mapsto l_{j_{q+1}}'\}] 
                \\
                &\text{Disjoint domains and reordering}
                \\
                =\:& s_i'\|s_{j_1}'\|\cdots\|s_{j_q}'\|s_{j_{q+1}}'\|s_{k_1}[s_i']\|\cdots\|s_{k_{\bar{m}}}[s_i']
                \\
                &\text{Definition of $s_{j_{q+1}}'$ and reordering}
            \end{align*}
    
            Next, recall that $(\forall q:s(\mathsf{loc}_q)\not \in K_q)\vee l_i\in K_i\vee(\exists j\in RS: l_j\in K_j)$ holds. By Lemma~\ref{lemma: K implies Comm(s)}, $(\forall q:s(\mathsf{loc}_q)\not \in K_q)$ implies $\neg Comm(s)$. By rule \textbf{ACT}, $l_i\in K_i\vee(\exists j\in RS: l_j\in K_j)$ implies that at least one of the transitions involved in the broadcast synchronization is committed, resulting in $b_i\vee b_{j_1}\vee\cdots\vee b_{j_m} = 1$. Hence, $Comm(s)\Rightarrow b$ always holds. By Definition~\ref{def: Restriction for TTSB}, the generated transition will be turned to a $\tau$-transition $s\xrightarrow{\tau}s'$ when we apply restriction $\backslash(\Delta\cup\mathcal{C})$ to $\mathsf{TTSB}(\mathcal{A}_1)\|\cdots\|\mathsf{TTSB}(\mathcal{A}_n)$.
            
           
            Finally, after applying $\mathsf{LTS}$, we obtain that the compositional semantics of $\mathcal{N}$
            contains transition $s\xrightarrow{\tau}s'$, as required.
        \end{itemize}
        \item $\supseteq$
        Assume $s\xrightarrow{\tau}s'$ in $\mathsf{LTS}((\mathsf{TTSB}(\mathcal{A}_1)\|\cdots\|\mathsf{TTSB}(\mathcal{A}_n))\backslash (\Delta\cup\mathcal{C}))$ is generated from a transition $s\xrightarrow{\delta!,b}s'$ with $Comm(s)\Rightarrow b$ in $\mathsf{TTSB}(\mathcal{A}_1)\|\cdots\|\mathsf{TTSB}(\mathcal{A}_n)$ during the process of restriction. According to Definition~\ref{def: Parallel composition}, the transition $s\xrightarrow{\delta!,b}s'$ must be generated from one $\delta!$-transition and $n-1$ $\delta?$-transitions. We denote them as $s_i\xrightarrow{\delta!,b_i}s_i'$, $s_{j_1}[s_i']\xrightarrow{\delta?,b_{j_1}}s_{j_1}',\dots,s_{j_m}[s_i']\xrightarrow{\delta?,b_{j_m}}s_{j_m}'$ and $s_{k_1}[s_i']\xrightarrow{\delta?,0}s_{k_1}[s_i'],\dots,s_{k_{\bar{m}}}[s_i']\xrightarrow{\delta?,0}s_{k_{\bar{m}}}[s_i']$ separately, where $n=1+m+\bar{m}$.
        By rules \textbf{ACT} and \textbf{VIRT}, the following expressions hold.
        \begin{align*}
            l_i\xrightarrow{g_i,\delta!,\rho_i}l_i' \quad s_i(\mathsf{loc}_i)=l_i \quad
            &s_i\models g_i \quad s_i'=\rho (s_i)[\{\mathsf{loc}_i\mapsto l_i'\}] \quad b_i \Leftrightarrow (l_i\in K_i)
            \\
            \forall j \in \{j_1,\dots,j_m\}: \quad
            &l_j\xrightarrow{g_j,\delta!,\rho_j}l_j' \myHspace s_j[s_i'](\mathsf{loc}_j)=l_j \myHspace s_j[s_i']\models g_j
            \\
            &s_j'=\rho (s_j[s_i'])[\{\mathsf{loc}_j\mapsto l_{j_1}'\}] \myHspace b_j \Leftrightarrow (l_j\in K_j)
            \\
            \forall k \in \{k_1,\dots,k_{\bar{m}}\}: \quad
            &s_k[s_i'](\mathsf{loc}_k)=l_k
            \quad \forall l_k\xrightarrow{g_k,\delta?,\rho_k} l_k':s_k[s_i']\not \models g_k
        \end{align*} 
        The transition $s\xrightarrow{\delta!,b}s'$ is constructed by using rule \textbf{RCV} $n-1$ times 
        and using rule \textbf{SND} once. We have $s'=s'_i\|s_{j_1}'\|\cdots\|s_{j_m}'\|s_{k_1}[s_i']\|\cdots$ $\|s_{k_{\bar{m}}}[s_i']$, and $b=b_i\vee b_{j_1}\vee\cdots\vee b_{j_m}$.
        
        Now, we prove that all the preconditions of rule \textbf{BCST} are satisfied.
        \begin{itemize}          
            \item 
            $s_i(\mathsf{loc}_i)=l_i \: \Rightarrow \: s(\mathsf{loc}_i)=l_i$
            \item 
            $s_i\models g_i \: \Rightarrow \: s\models g_i$
            \item 
            Since $\{j_1,\dots,j_m\} = RS, (\forall j\in\{j_1,\dots,j_m\} : l_j\xrightarrow{g_j,\delta!,\rho_j}l_j',s_j[s_i'](\mathsf{loc}_j)=l_j, (s_j[s_i']\models g_j)) \: \Rightarrow \: (\forall j\in RS:l_j\xrightarrow{g_j,\delta?,\rho_j}l_j',s(\mathsf{loc}_j)=l_j, s\models g_j)$
            \item
            Since $Comm(s)\Rightarrow{b}$ holds, we have the following derivation.
            \begin{samepage}
            \begin{align*}
                &Comm(s)\Rightarrow{b}
                \\
                \Leftrightarrow\:\:
                &\neg Comm(s) \vee b & \text{Equivalence}
                \\
                \Leftrightarrow\:\:
                &\neg Comm(s) \vee b_i\vee b_{j_1}\vee b_{j_2} \vee\cdots\vee b_{j_m} & \text{Substitution}
                \\
                \Leftrightarrow\:\:
                & \uline{\mathstrut (\neg Comm(s_1) \wedge\neg Comm(s_2) \wedge\cdots\wedge \neg Comm(s_n))} \vee
                \\
                & \uline{\mathstrut b_i} \vee  \uline{\mathstrut (b_{j_1}\vee b_{j_2}\vee\cdots\vee b_{j_m})} 
                & \text{Lemma~\ref{lemma: committed states}}
            \end{align*}
            \end{samepage}
            We analyze this disjunction in three parts marked by underscores.
            \begin{itemize}
                \item  By Lemma~\ref{lemma: K implies Comm(s)}, $\neg Comm(s_1) \wedge\cdots\wedge \neg Comm(s_n) \Leftrightarrow \forall q, s(\mathsf{loc}_q)\not\in K_q$.
                
                \item By rule \textbf{ACT}, $b_i\Leftrightarrow l_i\in K_i$.
    
                \item By definition of $RS$, we have $\{j_1,\dots,j_m\} \subseteq RS$. Then by rule \textbf{ACT}, $b_{j_1}\vee \dots\vee b_{j_m} \Rightarrow \exists j\in RS:l_j\in K_j$.
            \end{itemize}
            In summary, $(\forall q:s(\mathsf{loc}_q)\not \in K_q)\vee l_i\in K_i\vee(\exists j\in RS: l_j\in K_j)$ holds.
            \item 
            According to the previous proof, $s'= \rho_{j_m}(\cdots\rho_{j_1}(\rho_i(s)))[\{\mathsf{loc}_i\mapsto l_i',$ $\mathsf{loc}_{j_1}\mapsto l_{j_1}',\dots,\mathsf{loc}_{j_m}\mapsto l_{j_m}'\}]$ holds.
            
            With all the preconditions of rule \textbf{SND} of NTA semantics satisfied, the transition $s\xrightarrow{\tau}s'$ can be generated, as required.
        \end{itemize}
    \end{description}

\section{Proof of Theorem~\ref{theorem: precongruence}} \label{APP: proof of precongruence}
    Since $\mathcal{T}_1$ and $\mathcal{T}_2$ are comparable, $\mathcal{T}_1\|\mathcal{T}_3$ and $\mathcal{T}_2\|\mathcal{T}_3$ are comparable as well. Let $\mathcal{T}_{13}=\mathcal{T}_1\|\mathcal{T}_3$ and $\mathcal{T}_{23}=\mathcal{T}_2\|\mathcal{T}_3$. Let $\rm{Q}$ be a timed step simulation from $\mathcal{T}_1$ to $\mathcal{T}_2$. Define relation ${\rm{R}}\in S_{13}\times S_{23}$ by
    \begin{equation}
        r\|t \: {\rm{R}} \: s\|t' \Leftrightarrow (r \: {\rm{Q}}\: s \wedge t=t').
        \notag
    \end{equation}
    Let's prove that $\rm{R}$ is a timed step simulation from $\mathcal{T}_{13}$ to $\mathcal{T}_{23}$. Obviously, we have $(s^0_1\|s^0_2)\: {\rm{R}}\: (s^0_2\|s^0_3)$ since $s^0_1 \: {\rm{Q}} \: s^0_2$. Now, for any $(r\|t, s\|t)\in{\rm{R}}$, we prove that the four conditions in the definition of a timed step simulation are satisfied.  
    
    \begin{enumerate}
        \item 
        Since $r\:{\rm{Q}}\:s$, we have $r\lceil E_1 = s\lceil E_2$. This follows that $(r\|t)\lceil E_{13} = (s\|t)\lceil E_{23}$.
    
        \item 
        Pick $u\in V\!al(E_{13})$ and let $u'=u\lceil E_1$. Since ${\rm{Q}}$ is a timed step simulation, $r[u']\:{\rm{Q}}\:s[u']$. Since $\mathcal{T}_3$ is compatible with $\mathcal{T}_1$ and $\mathcal{T}_2$, $r[u'] = r[u]$ and $s[u']=s[u]$, $r[u]\:{\rm{Q}}\:s[u]$. Further by definition of ${\rm{R}}$, $r[u]\|t[u]\:{\rm{R}}\:s[u]\|t[u]$. Finally, by Lemma~\ref{lemma: expressions}(5), we have $(r\|t)[u]\:{\rm{R}}\:(s\|t)[u]$.
    
        \item 
        We derive
            \begin{align*}
                Comm(s\|t)
                \Rightarrow\:&
                Comm(s)\vee Comm(t) 
                &\text{Lemma~\ref{lemma: committed states}}
                \\
                \Rightarrow\:&
                Comm(r)\vee Comm(t)
                &\text{${\rm{Q}}$ a timed step simulation}
                \\
                \Rightarrow\:&
                Comm(r\|t) 
                &\text{Lemma~\ref{lemma: committed states}}
            \end{align*}
    
        \item 
        Assume that $\mathcal{T}_{13}$ has a transition $r\|t\xrightarrow{a,b}r'\|t'$, we prove that regardless of which rule in Fig.\ref{fig:Trans_Build_Rules} is used, either $\mathcal{T}_{23}$ has a transition $s\|t\xrightarrow{a,b}s'\|t''$ such that $r'\|t' \:{\rm{R}} \:s'\|t''$, or $a=\tau$ and $r'\|t' \:{\rm{R}}\: s\|t$. 
        \begin{itemize}
            \item Rules \textbf{EXT}, \textbf{TAU}, \textbf{SYNC} and \textbf{TIME}. The proof process is very similar to that in~\cite{Compositional_Abstraction}, except that the newly proved Lemma~\ref{lemma: committed states} is used.
        
            \item Rule \textbf{RCV}. 
            Then $\mathcal{T}_1$ and $\mathcal{T}_3$ separately have transitions $r\xrightarrow{\delta?,b}_1 r'$ and $t\xrightarrow{\delta?,b'}_3 t'$. These transitions generate $r\|t\xrightarrow{\delta?,b\vee b'}r'\|t'$. Since ${\rm{Q}}$ is a simulation, there exists a transition $s\xrightarrow{\delta?,b}_2 s'$ such that $r'\:{\rm{Q}}\:s'$. Let $t'' = t'$, then $s\|t\xrightarrow{\delta?,b\vee b'}s'\|t''$, and $r'\|t'\:{\rm{R}}\:s'\|t''$.
            
            \item Rule \textbf{SND} with $i=1$. 
            Then $\mathcal{T}_1$ and $\mathcal{T}_3$ separately have transitions $r\xrightarrow{\delta!,b}_1 r'$ and $t[r']\xrightarrow{\delta?,b'}_3 t'$. These transitions generate $r\|t\xrightarrow{\delta!,b\vee b'} r'\| t'$.
            Since ${\rm{Q}}$ is a simulation, there exists a state $s'$ such that $s\xrightarrow{\delta!,b}_2 s'$ and $r'\:\rm{Q}\:s'$. This follows that $r'\lceil E_1 = s'\lceil E_2$. Since $\mathcal{T}_3$ is compatible with $\mathcal{T}_1$ and $\mathcal{T}_2$, we have $t[r'] = t[s']$. Let $t'' = t'$, by rule \textbf{SND}, $\mathcal{T}_{23}$ has transition $s\|t \xrightarrow{\delta!,b\vee b'}s'\|t''$. Further by $r'\:\rm{Q}\:s'$, we have $r'\|t'\:{\rm{R}}\:s'\|t''$.
            
            \item Rule \textbf{SND} with $i=3$. 
            Then $\mathcal{T}_1$ and $\mathcal{T}_3$ separately have transitions $r[t']\xrightarrow{\delta?,b}_1 r'$ and $t\xrightarrow{\delta!,b'}_3 t'$. These transitions generate $r\|t\xrightarrow{\delta!,b\vee b'} r'\|t'$. Since $r\:{\rm{Q}}\:s$, ${\rm{Q}}$ is a simulation, and $\mathcal{T}_3$ is compatible with $\mathcal{T}_1$ and $\mathcal{T}_2$, then $r[t']\:{\rm{Q}}\:s[t']$ and $\mathcal{T}_2$ has a state $s'$ such that $s[t']\xrightarrow{\delta?,b}_2 s'$ and $r'\:\rm{Q}\:s'$. Let $t''=t'$, by rule \textbf{SND}, $\mathcal{T}_{23}$ has transition $s\|t \xrightarrow{\delta!,b\vee b'}s'\|t''$. Further by $r'\:\rm{Q}\:s'$, we have $r'\|t'\:{\rm{R}}\:s'\|t''$. 
        \end{itemize}
    \end{enumerate}

\section{Proof of Theorem~\ref{Theorem: side condition}} \label{APP: proof of side condition}

By Definition~\ref{def: Restriction for TTSB}, $\mathcal{T}_1 \backslash C$ and $\mathcal{T}_2 \backslash C$ remain comparable. Let $\rm{R}$ be the timed step simulation from $\mathcal{T}_1$ to $\mathcal{T}_2$, we prove it is still a timed step simulation from $\mathcal{T}_1 \backslash C$ to $\mathcal{T}_2 \backslash C$. For any state $r$ in $\mathcal{T}_1 \backslash C$ and $s$ in $\mathcal{T}_2 \backslash C$ where $r\:{\rm{R}}\:s$, it is obvious that the first and second conditions in Definition~\ref{def: Timed step simulation for TTSBs} still hold. Furthermore, from any committed state $r$ (resp. $s$) in $\mathcal{T}_1$ (resp. $\mathcal{T}_2$), there exists a committed transition from $r$ (resp. $s$) in $\mathcal{T}_1 \backslash C$ (resp. $\mathcal{T}_2 \backslash C$), which implies $Comm(s) \Rightarrow Comm(r)$. Finally, as the restriction operator makes the same modifications on the transition relations of $\mathcal{T}_1$ and $\mathcal{T}_2$, it is clear that the last condition also holds.

\section{Proof of Theorem~\ref{theorem: verify safety properties}} \label{APP: proof of verify safety properties}
    Suppose $\rm{Q}$ is a timed step simulation from $\mathcal{T}_a$ to $\mathcal{T}_b$, by Theorem~\ref{theorem: precongruence}, we can construct a relation ${\rm{R}}$ by
    \begin{equation}
        r\|t \: {\rm{R}} \: s\|t' \Leftrightarrow (r \: {\rm{Q}}\: s \wedge t=t')
        \notag
    \end{equation}
    which is a timed step simulation from $\mathcal{T}_a\|\mathcal{T}_c$ to $\mathcal{T}_b\|\mathcal{T}_c$. Then by Theorem~\ref{Theorem: side condition}, ${\rm{R}}$ is also a timed step simulation from $(\mathcal{T}_a\|\mathcal{T}_c)\backslash (\Delta\cup\mathcal{C})$ to $(\mathcal{T}_b\|\mathcal{T}_c) \backslash (\Delta\cup\mathcal{C})$. By Theorem~\ref{Theorem: semantics equal}, $\mathsf{LTS}(\mathcal{N})\cong\mathsf{LTS}((\mathcal{T}_a\|\mathcal{T}_c)\backslash (\Delta\cup\mathcal{C}))$ and $\mathsf{LTS}(\mathcal{N}')\cong\mathsf{LTS}((\mathcal{T}_b\|\mathcal{T}_c)\backslash (\Delta\cup\mathcal{C}))$. For any reachable state $r\|t$ of $\mathsf{LTS}(\mathcal{N})$, there exists at least one state $s\|t'$ of $\mathsf{LTS}(\mathcal{N}')$, such that $r\|t\: \rm{R}\: s\|t'$, i.e. $r \: \rm{Q}\: s \wedge t=t'$.  Since $\mathcal{T}_a$ and $\mathcal{T}_b$ are comparable, we have $E_a=E_b$, which implies $r\lceil{E_a} = s\lceil{E_b}$.     Since $P$ is a safety property of $\mathcal{N}'$, we have $s\|t' \models P$. Furthermore, since $P$ only depends on the variables of $\mathcal{T}_c$, it follows that $r\|t\models P$. In conclusion, $P$ is a safety property of $\mathcal{N}$.

\section{Details of Case Study I} \label{APP: Detail of Case Study I}
\subsection{Modeling the Producer-Consumer System}
The producer-consumer system is modeled by the NTA $\mathcal{N}=\langle \mathcal{A}_p,\mathcal{A}_c,\mathcal{A}_1,\dots,\mathcal{A}_\mathsf{N} \rangle$ shown in Fig.\ref{fig:CASE_STUDY}. The NTA consists of a producer $\mathcal{A}_p$ that produces data packets at fixed time intervals, $\mathsf{N}$ consumers $\mathcal{A}_1,\dots,\mathcal{A}_{\mathsf{N}}$ that process data packets, and a coordinator $\mathcal{A}_c$ to ensure the normal operations of the system. These components communicate via a shared variable, binary channels, and broadcast channels. 

\begin{figure}[h]
    \centering
    \begin{subfigure}{0.52\textwidth}
        \centering
        \includegraphics[width=\linewidth]{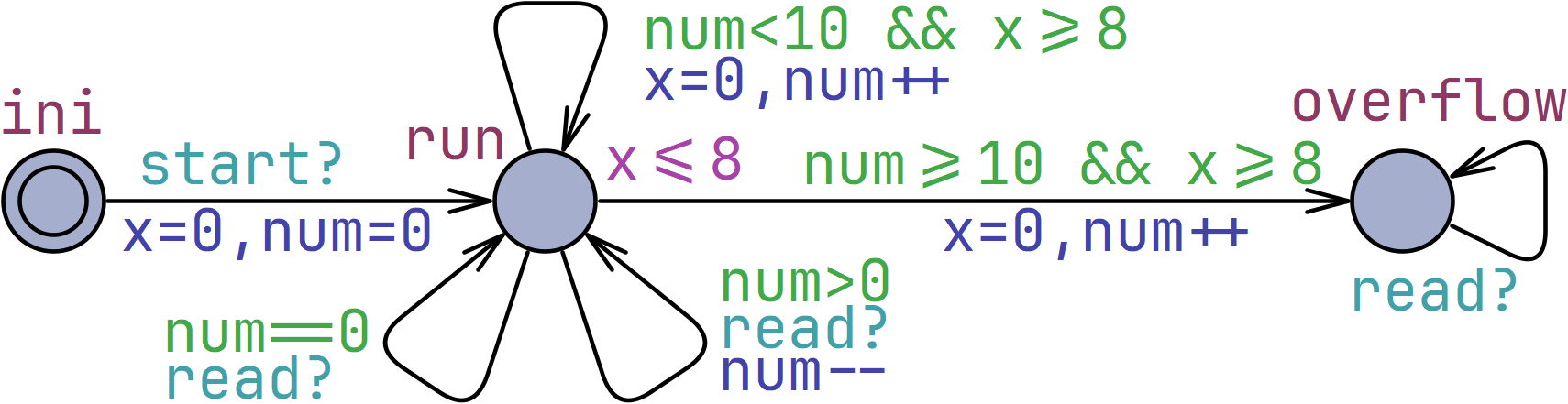}
        \caption{A producer $\mathcal{A}_p$}
        \label{fig:case_producer}
    \end{subfigure}
    
    \vspace{4mm}

    \begin{subfigure}{0.44\textwidth}
        \centering
        \includegraphics[width=\linewidth]{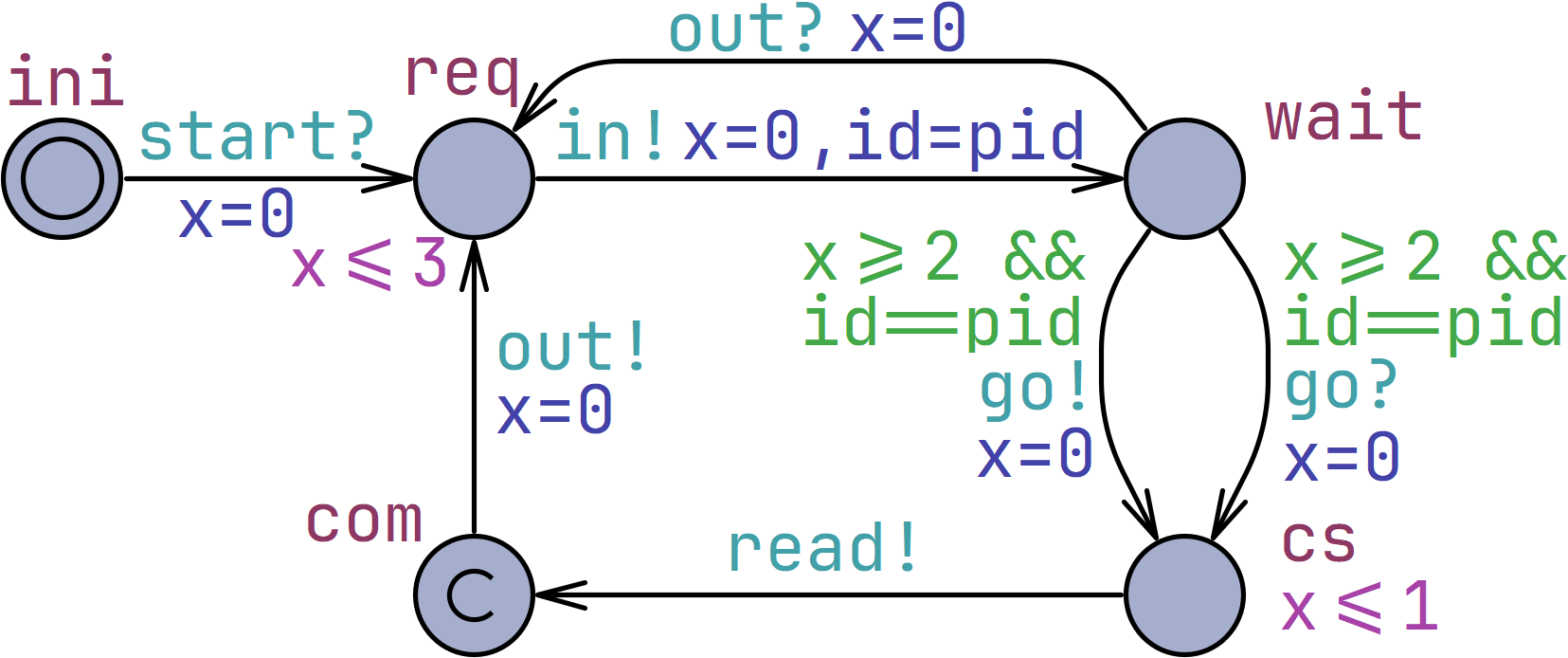}
        \caption{Consumers $\mathcal{A}_1\dots\mathcal{A}_n$}
        \label{fig:case_consumer}
    \end{subfigure}
    \hspace{6mm}
    \begin{subfigure}{0.29\textwidth}
        \centering
        \includegraphics[width=\linewidth]{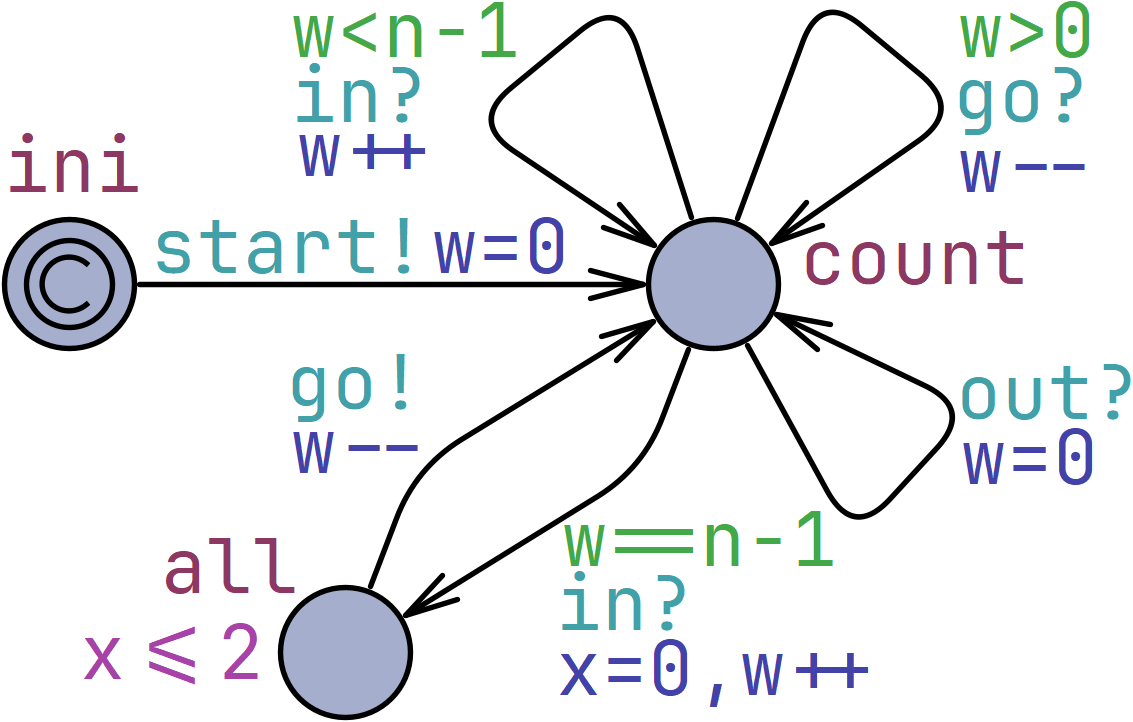}
        \caption{A coordinator $\mathcal{A}_c$}
        \label{fig:case_coordinator}
    \end{subfigure}
    
    \caption{A simple producer-consumer system}
    \label{fig:CASE_STUDY}
\end{figure}

The producer $\mathcal{A}_p$ generates a new data packet every $8$ time units recorded by its internal clock $x$. The integer variable $num$ tracks the number of data packets. If $num\neq 0$, the data packet can be consumed by some consumer through the binary channel $read$, with its value decreased by $1$. If the value exceeds $10$, the storage limit, the producer will transit to the $overflow$ location. 

Each consumer $\mathcal{A}_i$ has a unique $pid$ which values $i$ and a private clock $x$. When it receives the broadcast signal $start$ from the coordinator, implying the whole process has begun, it transits to its $req$ location. Then in $3$ time units, via binary action $in!$, the consumer notifies the coordinator that it has a request to read a data packet and then moves to the $wait$ location. Correspondingly, the shared variable $id$ that records which consumer has the latest request, is updated by $pid$ ($id=pid$). 
After waiting at least $2$ time units in $wait$, if there is no other request, i.e., $id$ still equals to the $pid$ of $\mathcal{A}_i$, the consumer may inform the coordinator that it enters the $cs$ location via the binary action $go!$ or be forced by the coordinator to enter the $cs$ location through the $go$ signal. 
In the $cs$ location, the consumer processes data within $1$ time unit. It then notifies the producer via the binary action $read!$ that a data packet has been consumed. Finally, the consumer broadcasts the action $out!$, notifying the waiting consumers to start the next round of data packet requests.

As mentioned, the coordinator starts all components simultaneously by broadcasting $start!$. 
It uses a variable $w$ to track the number of consumers that are in the $wait$ location. The value of $w$ is incremented or decremented via the upper transitions when a consumer enters or leaves the $wait$ location, triggered by the actions $in!$ or $go!$. When all consumers reach the $wait$ location, implying that no other consumer can update $id$, it might occur that all the consumers stay in this location infinitely. To avoid this, a lower-left loop operation of the coordinator is used to force the consumer whose $pid$ equals $id$ to enter $cs$ location for packet consumption. When the coordinator receives the $out$ signal from a consumer, it resets $w$ to $0$ since all the waiting consumers return to the $req$ location.

The objective is to verify that the producer will not overflow, i.e., $P = \{\mathsf{loc}_p \not= overflow\}$ is a safety property of $\mathcal{N}$. As Table~\ref{tab:example_PCS} shows, the verification time required by the traditional monolithic method grows exponentially as $n$ increases and exceeds $3,600$ seconds when $\mathsf{N}=15$ on our experimental equipment.

\subsection{Compositional Verification} 
To apply our compositional verification, we abstract the coordinator and all consumers in the system into the automaton $\mathcal{A}$ shown in Fig~\ref{fig:The first abstraction}.
We do not make abstraction on $\mathcal{A}_p$ due to requirements of Theorem~\ref{theorem: verify safety properties}, that is, property $P$ depends on the location of $\mathcal{A}_p$. Let $\mathcal{T} = \mathsf{TTSB}(\mathcal{A})\backslash\{in,go,out\}$ and $\mathcal{T}_r = (\mathcal{T}_c\|\mathcal{T}_1\|\dots\|\mathcal{T}_{\mathsf{N}})\backslash \{in,go,out\}$, where $\mathcal{T}_c, \mathcal{T}_1,\dots, \mathcal{T}_\mathsf{N}$ are the corresponding TTSBs of $\mathcal{A}_c, \mathcal{A}_1, \dots, \mathcal{A}_{\mathsf{N}}$. Below, we provide the proof for $\mathcal{T}_r\preceq\mathcal{T}$, and for brevity in the proof, we only consider the case where $n > 1$.

\begin{figure}[h]
    \vspace{-2mm}
    \centering
    \includegraphics[height=1.5cm]
    {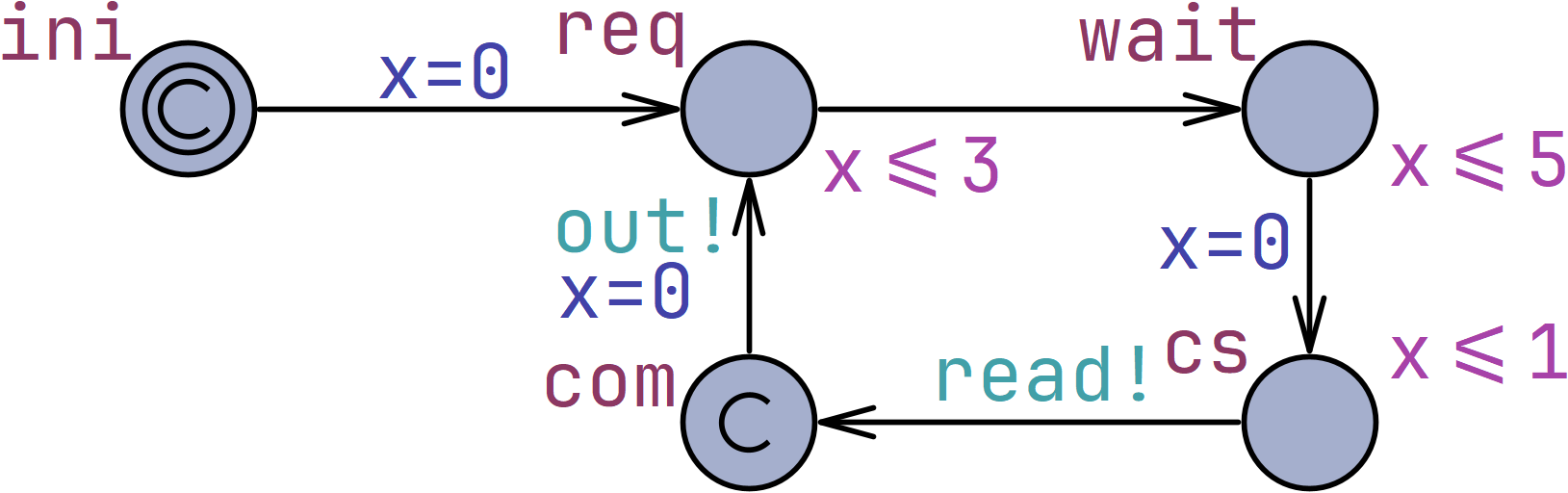}
    \vspace{-1mm}
    \caption{Abstraction $\mathcal{A}$}
    \label{fig:The first abstraction}
    \vspace{-6mm}
\end{figure}

\begin{proof}
For brevity in the proof, we only consider the case where $n > 1$. 
$\mathcal{T}_r$ and $\mathcal{T}$ are comparable since neither of them has external variables. To construct the $\rm{R}$ relation that satisfies the requirements in Definition~\ref{def: Timed step simulation for TTSBs}, we classify the states in $\mathcal{T}_r$ into 6 categories and find corresponding states in $\mathcal{T}$.


\paragraph{\bf Category 0: The initial state.}
According to Definitions~\ref{def: TTSB semantics of TA} and~\ref{def: Parallel composition}, the initial states of $\mathcal{T}_r$ and $\mathcal{T}$ are $s_r^0=\{\mathsf{loc}_c\mapsto ini,x_c\mapsto 0,w\mapsto 0,id\mapsto 0,\mathsf{loc}_1\mapsto ini,$ $x_1\mapsto 0,\dots,\mathsf{loc}_{\mathsf{N}}\mapsto ini,x_{\mathsf{N}}\mapsto 0 \}$ and $s^0=\{\mathsf{loc}\mapsto ini,x\mapsto 0\}$, respectively. So we include $\langle s_r^0,s^0\rangle$ in the $\rm{R}$ relation. Since neither $\mathcal{T}_r$ nor $\mathcal{T}$ contain external variables, conditions $1$ and $2$ in Definition~\ref{def: Timed step simulation for TTSBs} are satisfied and will not be repeated in the subsequent discussion. According to Definition~\ref{def: TTSB semantics of TA} and Lemma~\ref{lemma: committed states}, both $s^0$ and $s_r^0$ are committed, so $Comm(s^0)\rightarrow Comm(s_r^0)$ holds. 

\paragraph{\bf Category 1: All consumers are at the $\bm{req}$ location.}
After the broadcast synchronization through the channel $start$, $\mathcal{T}_r$ reaches state $s_r^0[\{\mathsf{loc}_c\mapsto count,$ $\mathsf{loc}_1\mapsto req,\dots,\mathsf{loc}_{\mathsf{N}}\mapsto req\}]$. We categorize the $\mathcal{T}_r$ states that satisfy both following conditions into one category and represent its elements with $s_r^1$. 

\begin{enumerate}
    \item $\forall\: 1\leq i \leq n : s_r^1(\mathsf{loc}_i) = req$.
    \item $s_r^1(x_1) =s_r^1(x_2) =\dots= s_r^1(x_{\mathsf{N}})$.
\end{enumerate}

According to the invariants of the $req$ locations in $\mathcal{A},\mathcal{A}_1,\dots,\mathcal{A}_{\mathsf{N}}$, for each $s_r^1$, by letting $s^0$ undergo the same transition sequence from $s_r^0$ to $s_r^1$, we can always get the corresponding state $s^1$. We put all these pairs $\langle s_r^1, s^1 \rangle$ into $\rm{R}$.
    

\paragraph{\bf Category 2: Some consumer reaches the $\bm{wait}$ location.}   
After a $\tau$-transition from $s_r^1$, $\mathcal{T}_r$ reaches state $s_r^1[\{w\mapsto 1,\mathsf{loc}_i\mapsto wait,x_i\mapsto 0\}]$, where $1\leq i\leq n$. If a state satisfies the following conditions, we denote it as $s_r^2$.

\begin{samepage}
\begin{enumerate}
    \item $\forall\: 1\leq i \leq n : s_r^2(\mathsf{loc}_i) = req\vee s_r^2(\mathsf{loc}_i) = wait$.
    \item $\exists\: 1\leq i \leq n : s_r^2(\mathsf{loc}_i) = wait$.
\end{enumerate}
\end{samepage}
For each $s_r^2$, the corresponding state of $\mathcal{T}$ can be denoted as $s^2 = \{\mathsf{loc} \mapsto wait,$ $ x \mapsto \max_{1\leq i \leq n} s_r^2(x_i)\}$. Although the consumers' $wait$ location has no invariants, the coordinator imposes an upper limit on the time passage allowed from $s_r^2$. Combined with $I(req) = {x_i \leq 3}$ in $\mathcal{A}_1,\dots,\mathcal{A}_{\mathsf{N}}$, $I(full) = {x_c \leq 2}$ in $\mathcal{A}_c$ and $I(wait) = {x \leq 5}$ in $\mathcal{A}$, such $s^2$ can be reached through a $\tau$-transition and several time-passages from $s^1$. We add all these pairs $\langle s_r^2, s^2\rangle$ to $\rm{R}$. 


\paragraph{\bf Category 3: One consumer reaches the $\bm{cs}$ location.}
After a $\tau$-transition from $s_r^2$, $\mathcal{T}_r$ reaches the state $s_r^2[\{\mathsf{loc}_c\mapsto count, w\mapsto s_r^2(w)-1, \mathsf{loc}_j\mapsto cs,$ $x_j\mapsto 0\}]$. We denote the state that satisfies the following conditions as $s_r^3$.

\begin{enumerate}
    \item $\exists\: 1\leq i \leq n : s_r^3(\mathsf{loc}_i) = cs$.
    \item $\forall\: j \not = i: s_r^3(\mathsf{loc}_j) \not = cs$.
\end{enumerate}
For any $s_r^3$, let $s^3=\{\mathsf{loc}\mapsto cs, x\mapsto s_r^3(x_i)\}$, which can be reached from $s^2$ through a $\tau$-transition and several time-passages. We include all these pairs $\langle s_r^3, s^3\rangle$ in the $\rm{R}$ relation. It is worth noting that the state $s_r^3[\{\mathsf{loc}_j \mapsto cs, x_j\mapsto 0\}]$, where $j\not = i$, is not reachable from any $s_r^3$ state. Since $\mathcal{A}_i$ can remain in its $cs$ location for at most $1$ time unit, while other consumers must first set $id$ to their own $pid$ and then wait at least $2$ time units after $\mathcal{A}_i$ enters $cs$ before they can enter their own $cs$ locations.


\paragraph{\bf Category 4: One consumer reaches the $\bm{com}$ location.}
After a $read!$-transition for $s_r^3$, $\mathcal{T}_r$ reaches state $s_r^4=s_r^3[\{\mathsf{loc}_i\mapsto com\}]$. We let $s^3$ go correspondingly through $read!$-transition and reaches state $s^4=s^3[\{\mathsf{loc} \mapsto com\}]$. Obviously, $Comm(s^4)\Rightarrow Comm(s_r^4)$ holds and we include $\langle s_r^4, s^4\rangle$ in $\rm{R}$.

\paragraph{\bf Category 5: One consumer leaves the $\bm{wait}$ location.}   
Starting from state $s_r^4$, $\mathcal{T}_r$ go through an $out!$-transition and reaches $s_r^5 = s_r^4[\{w\mapsto 0, \mathsf{loc}_i\mapsto req,$ $x_i \mapsto 0, \forall j,s_r^4(\mathsf{loc}_j)=wait: \mathsf{loc}_j\mapsto req, x_j\mapsto 0\}]$. Correspondingly, we let $\mathcal{T}$ go through $out!$-transition from $s^4$ to $s^5=\{\mathsf{loc}\mapsto req, x\mapsto 0\}$, and include $\langle s_r^5, s^5\rangle$ in the $\rm{R}$ relation. According to the fact that $I(req)=\{x\leq 3\}$ for all $\mathcal{A},\mathcal{A}_1,\dots,\mathcal{A}_{\mathsf{N}}$, for any $s_r^5\oplus d$, where $d\leq 3-\max_{1\leq i \leq n} s_r^5(x_i)$, there is a corresponding $s^5\oplus d$. We include all these $\langle s_r^5\oplus d, s^5\oplus d \rangle$ in the $\rm{R}$ relation.

\vspace{2mm}

If $\mathcal{T}_r$ goes through a $\tau$-transition from $s_r^5\oplus d$ to $(s_r^5\oplus d)[\{\mathsf{loc}_i \mapsto wait, x_i \mapsto 0\}]$, where $1\leq i \leq n$, it can be observed that this newly reached state can be included to \textbf{Category 2}, which we have discussed previously. Thus, we have proven that $\rm{R}$ is a timed step simulation from $\mathcal{T}_r$ to $\mathcal{T}$,  i.e. $\mathcal{T}_r \preceq \mathcal{T}$. $\hfill\square$

\end{proof}

Let $\mathcal{N'}=\langle \mathcal{A},\mathcal{A}_p\rangle$, we can obtain that $P$ is a safety property of $\mathcal{N'}$ within $5$ millisecond with U{\scriptsize PPAAL}. At the same time, it is not difficult to see that $\mathcal{T}_r\|\mathsf{TTSB}(\mathcal{A}_p)$ satisfies the side condition of Theorem~\ref{Theorem: side condition}. Then by Theorem~\ref{theorem: verify safety properties}, we conclude that $P$ is a safety property of the original system $\mathcal{N}$. 


\section{Details of Case Study II} \label{APP: Detail of Case Study II}
\subsection{Modeling the Clock Synchronization Protocol~\cite{new_case_try}}
Assume there is a finite, fixed set of wireless nodes $\mathsf{Nodes} = \{0,\dots, \mathsf{N}-1\}$. As shown in Fig~\ref{fig:case_WSN}, each individual node $i\in\mathsf{Nodes}$ is described by three timed automata: $\mathcal{A}_{\textbf{C}[i]}$, $\mathcal{A}_{\textbf{W}[i]}$, and $\mathcal{A}_{\textbf{S}[i]}$. Automaton $\mathcal{A}_{\textbf{C}[i]}$ models the node's hardware clock, which may drift, automaton $\mathcal{A}_{\textbf{W}[i]}$ takes care of broadcasting messages, and automaton $\mathcal{A}_{\textbf{C}[i]}$ resynchronizes the hardware clock upon receipt of a message. 
The model constructed in~\cite{new_case_try} of the clock synchronization protocol is an NTA $\mathcal{N}$ consisting of timed automata $\mathcal{A}_{\textbf{C}[i]}$, $\mathcal{A}_{\textbf{W}[i]}$ and $\mathcal{A}_{\textbf{S}[i]}$, for each $i\in\mathsf{Nodes}$.

\begin{figure}[htbp]
    \vspace{-3mm}
    \centering
    \begin{minipage}[b]{0.38\textwidth}
        \centering
        \includegraphics[width=\textwidth]{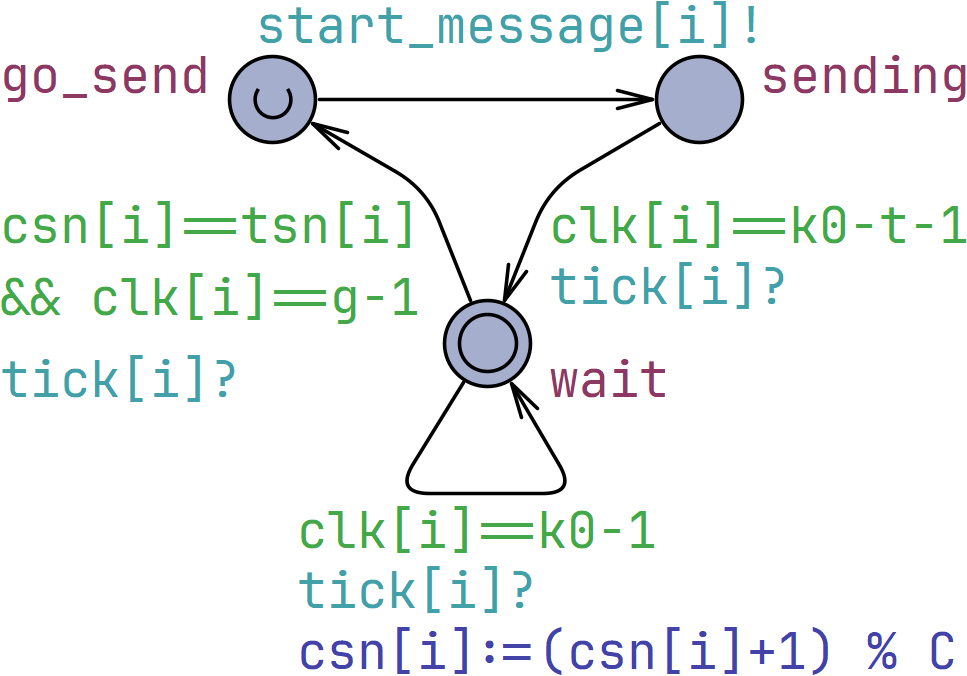}
        \subcaption*{(b) $\mathcal{A}_{\textbf{W}[i]}$}\label{fig:WSN}
    \end{minipage}
    \hspace{0.3cm}
    \begin{minipage}[b]{0.37\textwidth}
        \centering
        \begin{minipage}[b]{\textwidth}
            \centering
            \includegraphics[width=\textwidth]{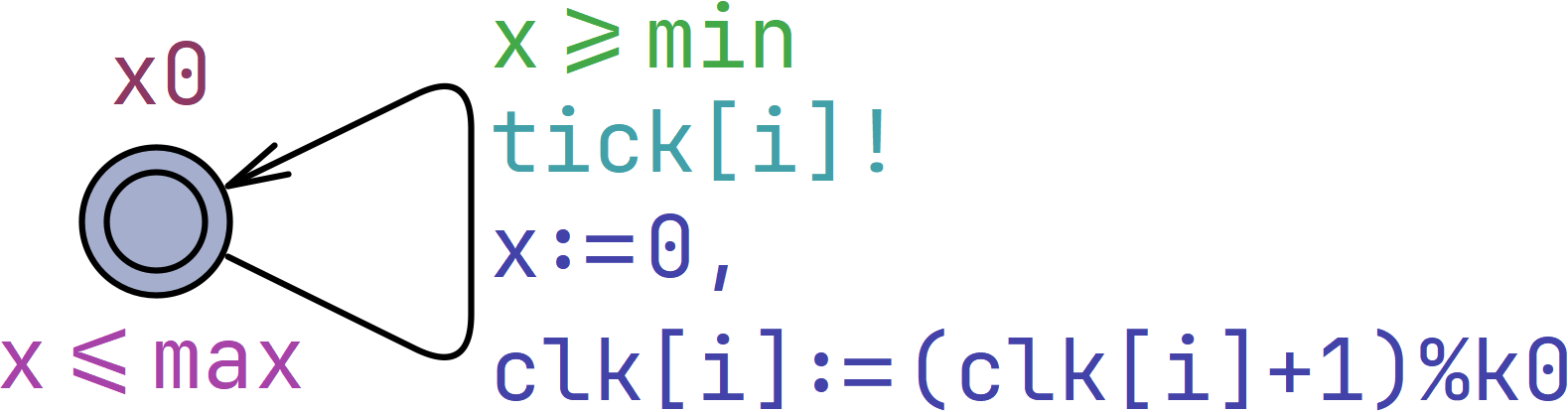}
            \subcaption*{(a) $\mathcal{A}_{\textbf{C}[i]}$}\label{fig:CLOCK}
        \end{minipage}
        \\[0.2cm] 
        \begin{minipage}[b]{0.74\textwidth}
            \centering
            \includegraphics[width=\textwidth]{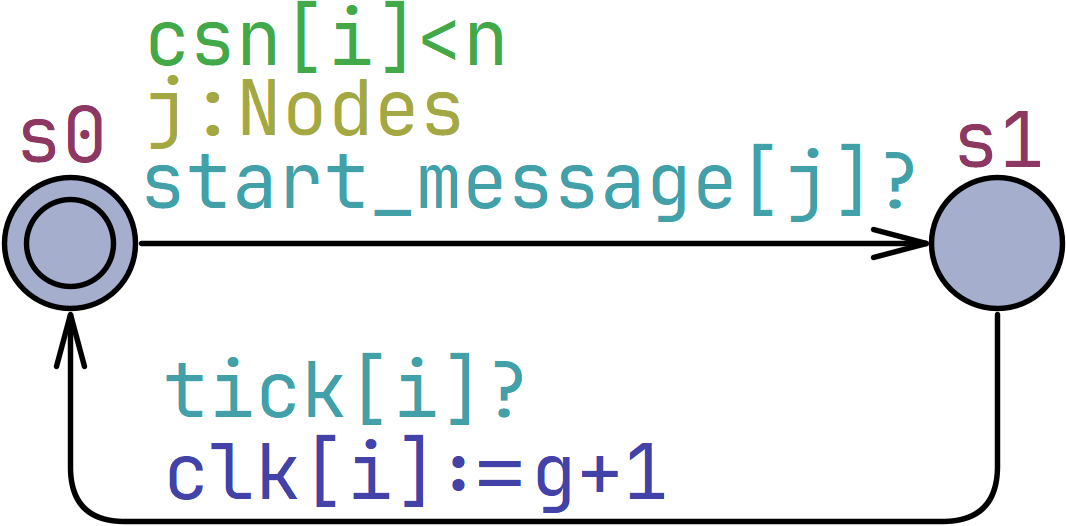}
            \subcaption*{(c) $\mathcal{A}_{\textbf{S}[i]}$}\label{fig:SYNC}
        \end{minipage}
    \end{minipage}
    \caption{Automata in a node}
    \label{fig:case_WSN}
    \vspace{-3mm}
\end{figure}

Timed automaton $\mathcal{A}_{\textbf{C}[i]}$, as shown in Fig.~\ref{fig:case_WSN}(a), models the behavior of the hardware clock of node $i$. It has a single location and a single transition. It has a private clock $x$, which measures the time between clock ticks. Whenever $x$ reaches the value $min$, $\mathcal{A}_{\textbf{C}[i]}$ can broadcast the action $tick[i]!$. The $tick[i]!$ action must occur before $x$ reaches value $max$. Then $x$ is reset to $0$, and the value of the node's hardware clock $clk[i]$ is incremented by $1$. The hardware clock is reset after $k_0$ ticks, i.e., the $clk[i]$ takes integer values modulo $k_0$, where $k_0$ is the number of clock ticks each \emph{slot} has. In brief, the operation time of the network is divided into fixed-length \emph{frames}, and each frame is subdivided into $\mathsf{C}$ slots. For more details about the frame and slot, please refer to~\cite{new_case_try}.

The first $n$ slots of each frame can be used for clock synchronization among the nodes. For a node $i$, it can broadcast messages in the ${tsn[i]}^{th}$ slot within a frame, where $tsn[i]$ is a constant. At the beginning of this slot, before clock synchronization, node $i$ should wait for $g$ ticks for other nodes to be ready to receive messages, and similarly, node $i$ should also wait for $t$ ticks at the end of this slot.
Following this, automaton $\mathcal{A}_{\textbf{W}[i]}$ it built, as shown in Fig.~\ref{fig:case_WSN}(b). It has three locations, four transitions, and an integer-type external variable, $csn[i]$, which records its current slot number. 
$\mathcal{A}_{\textbf{W}[i]}$ stays in the initial location $wait$ until the current slot number equals $tsn[i]$, and the $g^{th}$ clock tick in this slot occurs. It then transits to the urgent location $go\_send$, immediately broadcasts the action $start\_message[i]!$, and further transits to location $sending$. The broadcast channel $start\_message[i]$ is used to inform all neighboring nodes that a new message transmission has started. The automaton stays in location $sending$ until the start of the tail interval, that is, until the $(k_0-t)^{th}$ tick in the current slot, and then returns to location $wait$. At the end of each slot, that is, when the ${k_0}^{th}$ tick occurs, the automaton increments its current slot number by 1 modulo $\mathsf{C}$. Recall that $\mathsf{C}$ is the number of slots each frame has.

Automaton $\mathcal{A}_{\textbf{S}[i]}$, shown in Fig.~\ref{fig:case_WSN}(c), performs the role of the clock synchronizer. It has two locations and two transitions. The automaton waits in its initial location $s_0$ until it detects a new message, that is, until a $start\_message[j]?$ action occurs for some $j$. Here, the U{\scriptsize PPAAL} select statement, $j:\mathsf{Nodes}$, represents the non-deterministically selected $j \in \mathsf{Nodes}$. The automaton then moves to location $s_1$, provided the current slot can be used for clock synchronization, that is, $csn[i] < n$. Considering that when the $start\_message[j]?$ action occurs, the hardware clock of node $j$, $clk[j]$ values $g$. Therefore, node $i$ resets its own hardware clock $clk[i]$ to $g + 1$ upon occurrence of the first clock tick following $start\_message[j]?$. The automaton then returns to $s_0$.

As can be seen, in the designed model, the broadcast channel $tick[i]$ and shared variables $clk[i]$ and $csn[i]$ are only used for communication in each node $i$, while the broadcast channel $start\_message[i]$ is used for communication among nodes, where $i\in \mathsf{Nodes}$. So, we can restrict the communication capabilities between the nodes by modifying the select statement on the transition in $\mathcal{A}_{\textbf{S}[i]}$. Here, we allow all nodes in the network to communicate with each other, i.e., the network topology is fully connected. For the parameters in the model, we select $min=30$, $max=31$, $k_0=21$, $g=10$, $t=2$, $\mathsf{C}=\mathsf{N}+2$ and $n=\mathsf{N}$. Formally, the target property is $P=\{\forall i,j\in \mathsf{Nodes}: \mathsf{loc} (\mathcal{A}_{\textbf{W}[i]})\mapsto sending \Rightarrow csn[i] = csn[j]$\}. Table~\ref{tab:example_WSN} shows that the verification time required by the traditional monolithic method grows significantly as $\mathsf{N}$ increases and exceeds $3,600$ seconds when $\mathsf{N}=6$ on our experimental equipment. 

\subsection{Compositional Verification}
To apply our compositional verification, firstly, for each node $i\in \mathsf{Nodes}$, we first associate a TTSB $\mathcal{T}_{\textbf{N}[i]}$ with its NTA, where

\begin{equation*}
    \mathcal{T}_{\textbf{N}[i]}=\mathsf{TTSB}(\mathcal{A}_{\textbf{C}[i]})\|\mathsf{TTSB}(\mathcal{A}_{\textbf{W}[i]})\|\mathsf{TTSB}(\mathcal{A}_{\textbf{S}[i]})
\end{equation*}
Note that when considering $\mathsf{TTSB}(\mathcal{A}_{\textbf{W}[i]})$, we can replace the urgent location in $\mathcal{A}_{\textbf{W}[i]}$ with an ordinary one by introducing an additional clock. 

Then we select two nodes, $a$ and $b$, from the set $\mathsf{Nodes}$ with $0\le a <b \le \mathsf{N}-1$, and the TTSB associated with the remaining $\mathsf{N}-2$ nodes is


\begin{equation*}
    \mathcal{T}_{\mathsf{N}-\{a,b\}}={\Big\|}_{i\in\mathsf{Nodes}-\{a,b\}} \mathcal{T}_{\textbf{N}[i]}\backslash\{tick[i],clk[i],csn[i]\}
\end{equation*}
As the broadcast channel $tick[i]$ and shared variables $clk[i]$, $csn[i]$ are only used for communication within node $i$ but not used between nodes, so each $\mathcal{T}_{\textbf{N}[i]}$ with $i\in\mathsf{Nodes}-\{a,b\}$ should be  restricted by the set $\{tick[i],clk[i],csn[i]\}$. 

Now we need to construct the abstraction $\mathcal{A}$ with $\mathcal{T}_{\mathsf{N}-\{a,b\}}\preceq \mathsf{TTSB}(\mathcal{A})$. As observed, clock synchronization among nodes is fulfilled through the output broadcast action $start\_{message}[k]!$ with $k\in\mathsf{Nodes}$ in these cases:  
\begin{enumerate}
    \item occurrence of the $g^{th}$ clock tick since the start of the network.
    \item occurrence of the ${k_0}^{th}$ clock tick since last clock synchronization.
    \item occurrence of the $(\mathsf{C}-\mathsf{N}){k_0}^{th}$ clock tick since last clock synchronization.
\end{enumerate}
Considering these cases, we construct the TA $\mathcal{A}$ shown in Fig.~\ref{fig:abstraction_WSN}. It on one hand broadcasts $start\_{message}[i]!$ to $a$ and $b$, where $i\in\mathsf{abstractedNodes}=\mathsf{N}-\{a,b\}$, and on the other hand accepts $start\_{message}[j]$ from node $j$, where $j\in\mathsf{leftNodes}=\{a,b\}$.


\begin{figure}[h]
    \vspace{-2mm}
    \centering
    \includegraphics[width=0.95\textwidth]{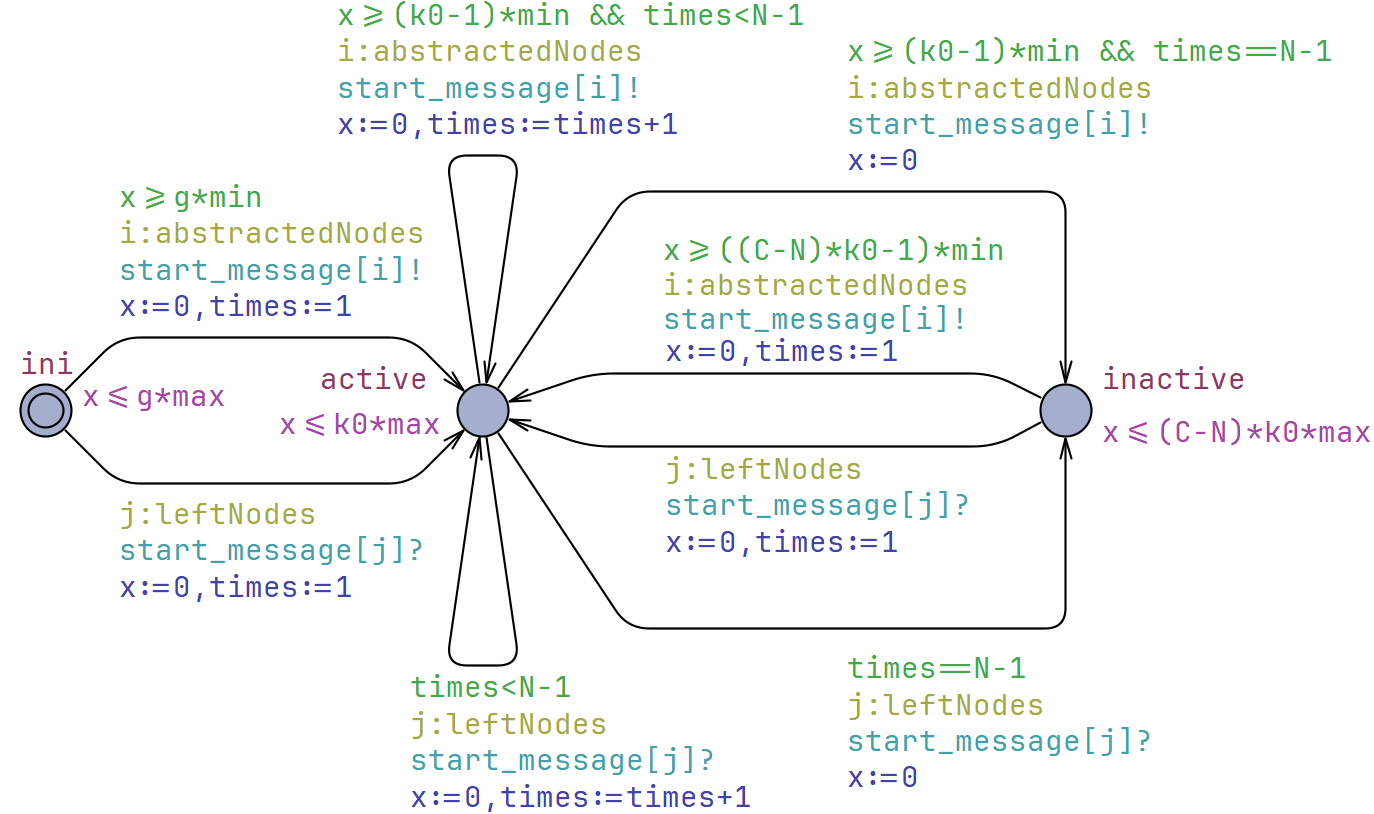}
    \caption{Abstraction $\mathcal{A}$}
    \label{fig:abstraction_WSN}
    \vspace{-2mm}
\end{figure}

Taking the same proof as that in Appendix~\ref{APP: Detail of Case Study I}, we conclude that $\mathcal{T}_{\mathsf{N}-\{a,b\}}\preceq \mathsf{TTSB}(\mathcal{A})$ holds for all $0\leq a < b\leq \mathsf{N}-1$. Then, the abstract NTA model is 

\begin{equation*}
\mathcal{N}_{\{a,b\}}=\langle \mathcal{A}_{\textbf{C}[a]}, \mathcal{A}_{\textbf{W}[a]}, \mathcal{A}_{\textbf{S}[a]}, \mathcal{A}_{\textbf{C}[b]}, \mathcal{A}_{\textbf{W}[b]}, \mathcal{A}_{\textbf{S}[b]},\mathcal{A} \rangle
\end{equation*}

Clearly, if we check that the hardware clocks of $a$ and $b$ remain synchronized during network operation for all the choices of $a$ and $b$, where $0\le a<b\le \mathsf{N}-1$, by Theorem~\ref{theorem: verify safety properties} and some simple logical transformations, we can conclude that $P$ is a safety property of the original NTA model. That is, 

\begin{equation*}
    \forall 0\le a<b\le \mathsf{N}-1: \mathcal{N}_{\{a,b\}}\models \forall \square P_{\{a,b\}} \:\:\Rightarrow\:\: \mathcal{N}\models \forall\square P
\end{equation*}
\begin{equation*}
P_{\{a,b\}}=\{(\mathsf{loc} (\mathcal{A}_{\textbf{W}[a]}) \mapsto sending \vee \mathsf{loc} (\mathcal{A}_{\textbf{W}[b]}) \mapsto sending) \Rightarrow csn[a] = csn[b]\}
\end{equation*}

Therefore, the total verification time for our compositional verification method is the sum of the checking times for all possible choices of $a$ and $b$. For instance, when $\mathsf{N}=6$, we need to verify $C_6^2=\frac{(6\times5)}{2} = 15$ different cases. The $\mathsf{CV}$ row of Table~\ref{tab:example_WSN} demonstrates the total verification time required by our compositional verification method for each $\mathsf{N}$. Although in the case of $\mathsf{N}=3$, our method takes a slightly longer time due to the fact that $\mathcal{A}$ has more behaviors than a single node, it demonstrates significant efficiency advantages when $\mathsf{N}\ge5$.

\end{document}